\documentclass[11pt, letterpaper]{amsart}
\usepackage{graphicx, amssymb, color}
\usepackage{amsmath}
\usepackage{enumerate}
\usepackage{dsfont}
\usepackage[round, sort&compress]{natbib}

\newtheorem{thm}{Theorem}[section]
\newtheorem{cor}[thm]{Corollary}
\newtheorem{lem}[thm]{Lemma}
\newtheorem{prop}[thm]{Proposition}
\theoremstyle{definition}
\newtheorem{defn}[thm]{Definition}

\newtheorem{ass}[thm]{Assumption}

\theoremstyle{remark}
\newtheorem{rem}[thm]{Remark}
\newtheorem{exa}[thm]{Example}

\numberwithin{equation}{section}

\newcommand{\Real}{\mathbb R}

\newcommand{\F}{\mathcal{F}}

\newcommand{\prob}{\mathbb{P}}
\newcommand{\expec}{\mathbb{E}}

\newcommand{\pare}[1]{\left(#1\right)}
\newcommand{\bra}[1]{\left[#1\right]}

\newcommand{\define}[1]{{\textbf{#1}}}
\newcommand{\nada}[1]{}

\newcommand{\BMO}{\texttt{BMO}}
\DeclareMathOperator*\esssup{ess\,sup}
\DeclareMathOperator*\essinf{ess\,inf}
\newcommand{\cA}{\mathcal{A}}
\newcommand{\C}{\mathcal{C}}
\newcommand{\cD}{\mathcal{D}}
\newcommand{\fF}{\mathfrak{F}}

\newcommand{\cY}{\mathcal{Y}}
\newcommand{\cZ}{\mathcal{Z}}

\newcommand{\cV}{\mathcal{V}}
\newcommand{\cU}{\mathcal{U}}
\newcommand{\bU}{\mathbb{U}}
\newcommand{\bV}{\mathbb{V}}
\newcommand{\cW}{\mathcal{W}}

\begin{document}

\title{Convex duality for stochastic differential utility}\thanks{The authors are grateful to Sara Biagini and Paolo Guasoni for inspiring discussions.}

\author[]{Anis Matoussi}
\address[Anis Matoussi]{CMAP, Ecole Polytechnique, Paris, and Universit\'{e} du Maine, Le Mans}
\email{anis.matoussi@univ-lemans.fr}

\author[]{Hao Xing}
\address[Hao Xing]{Department of Statistics,
London School of Economics and Political Science,
10 Houghton st,
London, WC2A 2AE,
UK}
\email{h.xing@lse.ac.uk}

\begin{abstract}
 This paper introduces a dual problem to study a continuous-time consumption and investment problem with incomplete markets and stochastic differential utility. For Epstein-Zin utility, duality between the primal and dual problems is established. Consequently the optimal strategy of the consumption and investment problem is identified without assuming several technical conditions on market model, utility specification, and agent's admissible strategy. Meanwhile the minimizer of the dual problem is identified as the utility gradient of the primal value and is economically interpreted as the ``least favorable" completion of the market.
\end{abstract}


\date{\today}

\maketitle

\section{Introduction}\label{sec: intro}
Classical asset pricing theory in the representation agent framework assumes that the representative agent's preference is modeled by a time-additive Von Neumann-Morgenstein utility. This specification restricts the relationship between risk aversion and intertemporal substitutability, leading to a rich literature of asset pricing anomalies, such as low risk premium and high risk-free rate. To disentangle risk aversion and intertemporal substitutability, the notion of recursive utility was introduced by \cite{Kreps-Porteus}, \cite{Epstein-Zin}, \cite{Weil}, amongst others. Its continuous-time analogue, stochastic differential utility, was defined by \cite{Epstein} for deterministic setting and \cite{Duffie-Epstein} in stochastic environment. The connection  between recursive utility and stochastic differential utility has also been rigorously established by  \cite{Kraft-Seifried} recently. Recursive utility and its continuous-time analogue generalize time-additive utility and provide a flexible framework to tackle aforementioned asset pricing anomalies, cf. \cite{Bansal-Yaron},  \cite{Bhamra-et-al}, \cite{Benzoni-et-al}, amongst others.

The asset pricing theory for recursive utility and stochastic differential utility builds on the optimal consumption and investment problems. For Epstein-Zin utility, a specification widely used in aforementioned asset pricing applications, its continuous-time optimal consumption and investment problems have been studied by \cite{Schroder-Skiadas-JET, Schroder-Skiadas-SPA},  \cite{Chacko-Viceira}, \cite{KSS-FS}, \cite{Kraft-Seiferling-Seifried}, and \cite{Xing}. These studies mainly utilize stochastic control techniques, either Hamilton-Jacobi-Bellman equation (HJB) in Markovian setting or backward stochastic differential equation (BSDE) in non-Markovian setting, to tackle the optimization problem directly. We call this class of methods \emph{primal approach}. However, the HJB equations rising from these problems are typically nonlinear and BSDEs are usually nonstandard. Therefore current available results obtained via primal approach still come with unsatisfactory restrictions on either market model, utility specification, or agent's admissible action.

In contrast, when portfolio optimization problems for time-additive utility are considered, a martingale (or duality) approach was introduced by  \cite{Pliska}, \cite{Cox-Huang}, \cite{Karatzas-Lehoczky-Shreve},  \cite{Karatzas-Lehoczky-Shreve-Xu}, \cite{He-Pearson}. Instead of tackling the primal optimization problem directly, a dual problem was introduced whose solution leads to the solution of the primal problem via the first order condition. This dual approach allows to strip away unnecessary assumptions and solve portfolio optimization problems with minimal assumptions on market model and utility, cf. \cite{Kramkov-Schachermayer-99, Kramkov-Schachermayer-03} for terminal consumption, \cite{Karatzas-Zitkovic} for intertemporal consumption.

This paper proposes a dual problem for an optimal consumption and investment problem in incomplete markets with stochastic differential utility. It is a minimization problem of a convex functional of state price densities (deflators). Similar to the primal problem, the dual value process aggregates the state price density and future evolution of the dual value process. Hence the dual problem also takes a recursive form, we call it \emph{stochastic differential dual}. Similar to time-additive utility, solution of this dual problem can be economically interpreted as the \emph{least favorable completion} of the market, i.e., the agent's optimal portfolio does not consist of the fictitious assets which are introduced to complete the market, cf. \cite{He-Pearson} and \cite{Karatzas-Lehoczky-Shreve-Xu}.

In contrast to time-additive utility, the convex functional appearing in the dual problem does not follow directly from applying Fenchel-Legendre transformation to the utility function. Instead we utilize a variational representation of recursive utility, introduced by \cite{Geoffard}, \cite{ElKaroui-Peng-Quenez} and \cite{Dumas-Uppal-Wang}, to transform the primal problem to a minmax problem, which leads to a variational representation of the dual problem. This dual variational representation can be transformed back to a recursive form, when the aggregrator of the recursive utility is homothetic in the consumption variable. Therefore this approach works particularly well for Epstein-Zin utility, without any assumption on risk aversion and elasticity of intertemporal substitution (EIS).

The dual problem gives rise to an inequality between the primal value function and the concave conjugate of the dual value function. When this inequality is an identity, there is \emph{duality} between primal and dual problems, or there is \emph{no duality gap}. Consider market models whose investment opportunities are driven by some state variables. We obtain duality in two situations: 1) non-Markovian models with bounded market price of risk, together with all possibilities on risk aversion and EIS whose associated Epstein-Zin utilities are known to exist; 2) Markovian models with unbounded market price of risk, including Heston model and Kim-Omberg model, when risk aversion and EIS are both in excess of one. This later market and utility specification are widely used in aforementioned asset pricing applications.

The duality between primal and dual problems allow us to simultaneously verify the primal and dual optimizers. On the primal side, technical conditions on utility and market model are removed. In particular, in contrast to the permissible class of strategies considered in \cite{Xing}, the primal optimality is established in the standard admissible class, which consists of all nonnegative self-financing wealth processes. On the dual side, the super-differential of the primal value is identified as the minimizer of the dual problem, extending this well known result from time-additive utility to stochastic differential utility. In the primal approach, super-differential of the primal value was mainly identified via the \emph{utility gradient} approach by \cite{Duffie-Skiadas}. In this approach, one needs to show that the sum of the deflated wealth process and integral of the deflated consumption stream is a martingale for candidate optimal strategy. This martingale property now becomes a direct consequence of duality.

The remaining of the paper is organized as fellows. After the dual problem is introduced for general stochastic differential utility in Section \ref{sec: dual}, it is specified to Epstein-Zin utility. The main results are presented in Section \ref{sec: no gap} where duality is established for two market and utility settings. In the second setting, we first introduce two abstract conditions which lead to duality. These abstract conditions are then specified as explicit parameter conditions in two examples. All proofs are postponed to appendix.

\section{Dual problem}\label{sec: dual}
\subsection{General setting}
 Let $(\Omega, (\F_t)_{0\leq t\leq T}, \F, \prob)$ be a filtrated probability space whose filtration $(\F_t)_{0\leq t\leq T}$ satisfies the usual assumptions of completeness and right-continuity. Let $\C$ be the class of nonnegative progressively measurable processes defined on $[0,T]$. For $c\in \C$, $c_t$, $t<T$, represents the consumption rate at time $t$ and $c_T$ stands for the lump sum bequest consumption at time $T$. The preference over $\C$-valued consumption stream is described by a stochastic differential utility, cf. \cite{Duffie-Epstein}.

 \begin{defn}\label{def: SDU}
  Given a bequest utility $U_T: \Real \rightarrow \Real$ and an \define{aggregator} $f: (0,\infty) \times \Real \rightarrow \Real$, a \define{stochastic differential utility} for $c\in \C$ is a semimartingale $(U^c_t)_{0\leq t\leq T}$ satisfying
  \begin{equation}\label{eq:sdu}
   U^c_t = \expec_t \Big[U_T(c_T) + \int_t^T f(c_s, U^c_s) ds\Big], \quad \text{ for all } t\leq T.
  \end{equation}
  Here $\expec_t[\cdot]$ stands for the conditional expectation $\expec[\cdot | \F_t]$.
 \end{defn}
 We assume that any utility process is an element of a class of processes $\cU$. Such a class will be specified in the next section when we focus on a specific class of stochastic differential utilities. For $c\in \C$, if the associated stochastic differential utility $U^c$ exists and $U^c\in \cU$, we call $c$ \emph{admissible} and denote the class of admissible consumption streams by $\C_a$.

 Consider a model of financial market with assets $S= (S^0, S^1, \dots, S^n)$, where $S^0$ is the price of a riskless asset, $(S^1, \cdots, S^n)$ are prices for risky assets, and $S$ is assumed to be a semimartingale whose components are all positive.

 An agent, starting with an initial capital $w>0$, invests in this market by choosing a portfolio represented by a predictable, $S$-integrable process $\pi= (\pi^0, \pi^1, \dots, \pi^n)$. With $\pi^i_t$ representing the proportion of current wealth invested in asset $i$ at time $t$, $\pi^0_t = 1-\sum_{i=1}^n \pi^i_t$ is the proportion invested in the riskless asset. Given  an investment strategy $\pi$ and a consumption stream $c$, agent's wealth process $\cW^{(\pi, c)}$ follows
 \begin{equation}
  d \cW^{(\pi, c)}_t = \cW^{(\pi, c)}_{t-} \pi^\top_t \frac{dS_t}{S_{t-}} - c_t dt, \quad \cW^{(\pi, c)}_0= w.
 \end{equation}
 A pair of investment strategy and consumption stream $(\pi, c)$ is \emph{admissible} if $c\in \C_a$ and $\cW^{(\pi, c)}$ is nonnegative. This restriction outlaws doubling strategies and ensures existence of the associated stochastic differential utility. The class of admissible pairs is denoted by $\cA$.

 The agent aims to maximize her stochastic differential utility at time $0$ over all admissible strategies, i.e.,
 \begin{equation}\label{eq:primal}
  U_0 = \sup_{(\pi, c)\in \cA}U^c_0.
 \end{equation}
 We call \eqref{eq:primal} the \define{primal problem}. When $U^c$ is concave in $c$, there is an associated \define{dual problem}. In order to formulate the dual problem,
 we focus on a class of stochastic differential utility whose aggregator satisfies the following assumption.
 \begin{ass}\label{ass:convex}
  $f(c, u)$ is \emph{concave} in $c$ and \emph{convex} in $u$.\footnote{The case where $f(c,u)$ is concave in $u$ can be treated similarity, see Remark \ref{rem:f-concave} below.
  The convexity (resp. concavity) of $f(c,u)$ in $u$ implies preference for early (resp. late) resolution of uncertainty (cf. \cite{Kreps-Porteus} and \cite{Skiadas}).}
 \end{ass}
 The previous assumption leads to an alternative characterization of stochastic differential utility. This so called \define{variational representation} was first proposed by \cite{Geoffard} in a deterministic continuous-time setting, and extended by \cite{ElKaroui-Peng-Quenez} and \cite{Dumas-Uppal-Wang} to uncertainty. Let us recall the \define{felicity function} $F$, defined as the Fenchel-Legendre transformation of $f$ with respect to its second argument:
 \begin{equation}\label{eq:F}
  F(c, \nu) := \inf_{u\in \Real} (f(c,u)+ \nu u).
 \end{equation}
 Convex analysis implies that $F(c, \nu)$ is concave in $\nu$, $f$ and $F$ satisfy the duality relation
 \begin{equation}\label{eq:dual-fF}
  f(c, u) = \sup_{\nu \in \Real} (F(c, \nu) - \nu u).
 \end{equation}
 Moreover, one can show that $F(c, \nu)$ is concave in $c$.

 For the variational representation, depending on the integrability of $c$ and $U^c$, certain integrability assumption on the dual variable $\nu$ is needed, for example, \cite{ElKaroui-Peng-Quenez} and \cite{Dumas-Uppal-Wang} consider square integrability $\nu$ when $c$ and $U^c$ are both square integrable. Rather than imposing specific integrability condition, we work with the following class of admissible dual variables, in order to allow for a wide class of utility processes.
 \begin{defn}\label{def: dualvar}
  For a progressively measurable process $\nu$, denote
  \[\kappa^\nu_{s,t} := \exp\Big(-\int_s^t \nu_u du\Big), \quad \text{for } s,t\in [0,T].\]
  $\nu$ is \emph{admissible} if
  \begin{enumerate}
   \item[(i)] $U_t^{c, \nu} := \expec_t \big[\kappa^\nu_{t, T} U_T(c_T) + \int_t^T \kappa^\nu_{t,s} F(c_s, \nu_s) ds\big]<\infty$ for any $t\in [0,T]$ and $c\in \C_a$.
   \item[(ii)] When $U_0^{c, \nu}>-\infty$, then $\kappa^\nu_{0, \cdot}U^{c, \nu}$ is of class (D)\footnote{A progressively measurable process $X$ is of class (D) if $\{X_\tau \,|\, \tau \text{ is finite stopping time}\}$ is uniformly integrable.}.
   \item[(iii)] $\kappa^\nu_{0, \cdot} U$ is of class (D) for any $U \in \cU$.
  \end{enumerate}
  The class of admissible $\nu$ is denoted by $\cV$.
 \end{defn}

 The following result is a minor extension of \cite[Section 3.2]{ElKaroui-Peng-Quenez} and \cite[Theorem 2.1]{Dumas-Uppal-Wang}.
 \begin{lem}\label{lem:var-rep}
  Let Assumption \ref{ass:convex} holds. For any $c\in \C_a$, the following statements hold:
  \begin{enumerate}
   \item[(i)] $U^c_0 \geq \sup_{\nu\in \cV} U_0^{c, \nu}$.
   \item[(ii)] If $\nu^c := - f_u(c, U^c) \in \cV$, then the inequality in part (i) is an identity.
  \end{enumerate}
 \end{lem}
 Let us now use the variational representation in Lemma \ref{lem:var-rep} part (ii) to derive the dual problem associated to \eqref{eq:primal}.
 When the assumption of Lemma \ref{lem:var-rep} part (ii) holds, the primal problem is transformed into
 \begin{equation}\label{eq:primal-rep}
 \begin{split}
  U_0 &= \sup_{(\pi, c)\in \cA} \sup_{\nu \in \cV} \expec\Big[\kappa^\nu_{0,T}U_T(c_T) + \int_0^T \kappa^\nu_{0,s} F(c_s, \nu_s) ds\Big]\\
  &= \sup_{\nu \in \cV} \sup_{(\pi, c)\in \cA} \expec\Big[\kappa^\nu_{0,T}U_T(c_T) + \int_0^T \kappa^\nu_{0,s} F(c_s, \nu_s) ds\Big].
 \end{split}
 \end{equation}
 For a given $\nu\in \cV$, the inner problem in the second line above can be considered as an optimization problem for a bequest utility $U_T$ and a time-additive intertemporal  utility $F(c, \nu)$, parameterized by $\nu$, which can be viewed as a fictitious discounting rate. To present the dual problem of this inner problem, we define the Fenchel-Legendre transform of $U_T$ and $F$ (with respect to its first argument):
 \begin{equation}\label{eq:V-G}
  V_T(d) := \sup_{c\in \Real} (U_T(c) - d\,c), \quad G(d, \nu) := \sup_{c\in \Real} (F(c, \nu) - d\, c),
 \end{equation}
 and recall the class of \define{state price densities} (\define{supermartingale deflators}):
 \[
  \cD := \{D \,|\, D_0=1, D\geq 0, D \cW^{(\pi, c)} + \int_0^\cdot D_s c_s ds \text{ is a supermartingale for all } (\pi, c)\in \cA\}.
 \]
 To exclude arbitrage opportunity, we assume
 \[
  \cD \neq \emptyset.
 \]
 Coming back to the second line in \eqref{eq:primal-rep}, using the dual problem of the inner problem, we obtain
 \begin{equation}\label{eq:dual-bdd}
 \begin{split}
  U_0 &\leq \sup_{\nu\in \cV} \inf_{y>0, D\in \cD}\left\{ \expec\Big[\kappa^\nu_{0,T} V_T((\kappa^\nu_{0,T})^{-1} y D_T) +\int_0^T \kappa^\nu_{0,s} G((\kappa^\nu_{0,s})^{-1} y D_s, \nu_s) ds\Big] + w \,y\right\}\\
  &\leq \inf_{y>0, D\in \cD} \sup_{\nu \in \cV} \left\{ \expec\Big[\kappa^\nu_{0,T} V_T((\kappa^\nu_{0,T})^{-1} y D_T) +\int_0^T \kappa^\nu_{0,s} G((\kappa^\nu_{0,s})^{-1} y D_s, \nu_s) ds\Big] + w \,y\right\}.
 \end{split}
 \end{equation}

Now the inner problem in the previous line can be viewed as a variational problem. In order to transform it back to a recursive form, we need to work with $U_T$ and $F$ which have the following homothetic property in $c$.
\begin{ass}\label{ass:homothetic}
 The bequest utility and the felicity function have representations
 \[
  U_T(c) = \tfrac{c^{1-\gamma}}{1-\gamma}\quad \text{and} \quad \quad F(c, \nu) = \tfrac{c^{1-\gamma}}{1-\gamma} F(\nu),
 \]
 where $1\neq \gamma>0$ is the \define{relative risk aversion} and $F$, overloaded with an univariate function, is positive on its effective domain and $\tfrac{F}{1-\gamma}$ is concave.
\end{ass}
The previous specification of $U_T$ and $F$ implies
\begin{equation}\label{eq:VG}
 V_T(d)= \tfrac{\gamma}{1-\gamma} d^{\tfrac{\gamma-1}{\gamma}} \quad \text{and} \quad G(d, \nu) = \tfrac{\gamma}{1-\gamma} d^{\tfrac{\gamma-1}{\gamma}} F(\nu)^{\tfrac{1}{\gamma}}.
\end{equation}
Come back to the second line in \eqref{eq:dual-bdd},
\begin{equation*}
\begin{split}
 &\kappa^\nu_{0,T} V_T((\kappa^\nu_{0,T})^{-1} y D_T) = (\kappa^\nu_{0,T})^{\tfrac{1}{\gamma}} V_T(yD_T) = \kappa^{\tfrac{\nu}{\gamma}}_{0,T} V_T(yD_T),\\
 &\kappa^\nu_{0,s} G((\kappa^\nu_{0,s})^{-1} y D_s, \nu_s) = (\kappa^\nu_{0,s})^{\tfrac{1}{\gamma}} G(y D_s, \nu_s) = \kappa^{\tfrac{\nu}{\gamma}}_{0,s} G(y D_s, \nu_s).
\end{split}
\end{equation*}
Combining the last two identities and the inner problem in the second line of \eqref{eq:dual-bdd}, we are motivated to introduce
\begin{equation}\label{eq:V^yDnu}
 V^{yD, \nu}_t := \expec_t \Big[\kappa^{\tfrac{\nu}{\gamma}}_{t, T} V_T(y D_T) + \int_t^T \kappa^{\tfrac{\nu}{\gamma}}_{t,s} G(y D_s, \nu_s) ds\Big].
\end{equation}
Therefore the second line in \eqref{eq:dual-bdd} takes the form
\begin{equation}\label{eq:dual ineq 0}
 U_0 \leq \inf_{y>0, D\in \cD} \sup_{\nu\in \cV} (V_0^{yD, \nu} + w y).
\end{equation}
To transfer this variational problem $\sup_{\nu\in \cV} V_0^{yD, \nu}$ back to a recursive form, we take Fenchel-Legendre transformation of $G$ with respect to its second variable, which requires the following
\begin{ass}\label{ass:G-concave}
 The function $G(d, \nu)$ is concave in $\nu$.\footnote{When the univariate function $F$ in Assumption \ref{ass:homothetic} is twice continuously differentiable, this assumption is equivalent to $(F')^2 + \tfrac{\gamma}{1-\gamma} F'' <0$.}
\end{ass}
This assumption allows us to define
\begin{equation}\label{eq:g}
 g(d, v) := \sup_{\nu \in \Real} (G(d, \nu) -\nu \,v),
\end{equation}
and introduce an analogue of stochastic differential utility for the dual problem.
\begin{defn}
 Suppose that $U_T$ and $F$ satisfy Assumptions \ref{ass:homothetic} and \ref{ass:G-concave}. A \define{stochastic differential dual} for $yD$ is a semimartingale $(V^{yD}_t)_{0\leq t\leq T}$ satisfying
 \begin{equation}\label{eq:sdd}
  V^{yD}_t = \expec_t \Big[V_T(yD_T) + \int_t^T g(y D_s, \tfrac{1}{\gamma} V^{yD}_s) ds\Big], \quad \text{ for all } t\leq T.
 \end{equation}
\end{defn}
Similar to stochastic differential utility, we denote by $\cD_a$ the class of state price density $D$ whose associated stochastic differential dual $V^{yD}$ exists for all $y>0$ and $V^{yD}\in \cU$. Moreover, we restrict $\cV$ such that $V_t^{yD, \nu}<\infty$ for any $D\in \cD_a, y>0$, $t\in [0,T]$, and $V^{yD, \nu}\in \cU$ when $V^{yD, \nu}>-\infty$. The same argument as in part (i) of Lemma \ref{lem:var-rep} then yields
\begin{lem}\label{lem:dual-val-rep}
 Let Assumptions \ref{ass:homothetic} and \ref{ass:G-concave} hold. For any $D \in \cD_a$ and $y>0$, we have $V_0^{yD}\geq \sup_{\nu\in \cV} V_0^{yD, \nu}$.
\end{lem}

As a result, for any $y>0$, we call the following problem the \define{dual problem} of \eqref{eq:primal}.
\begin{equation}\label{eq:dual}
 V_0^y = \inf_{D\in \cD_a} V_0^{yD}.
\end{equation}
A diagram illustrating relationship between various functions introduced above is presented in Figure \ref{fig:transformation}, starting from the primal problem in the upper left corner and ending at the dual problem in the bottom left corner.
Combining \eqref{eq:primal-rep}, \eqref{eq:dual-bdd}, \eqref{eq:dual ineq 0}, and Lemma \ref{lem:dual-val-rep}, we now obtain the following inequality between primal and dual problems.

\begin{prop}\label{prop:duality-bd}
 Let Assumptions \ref{ass:convex}, \ref{ass:homothetic}, and \ref{ass:G-concave} hold, moreover, $\nu^c = -f_u(c, U^c)\in \cV$ for any $c\in \C_a$. Then
 \begin{equation}\label{eq:dual-ineq}
  \sup_{(\pi, c)\in \cA} U^c_0 \leq \inf_{y>0} (\inf_{D\in \cD_a}V^{yD}_0 + w \,y).
 \end{equation}
\end{prop}

\begin{figure}[h]\label{fig:transformation}
\centering
\includegraphics[scale=0.55]{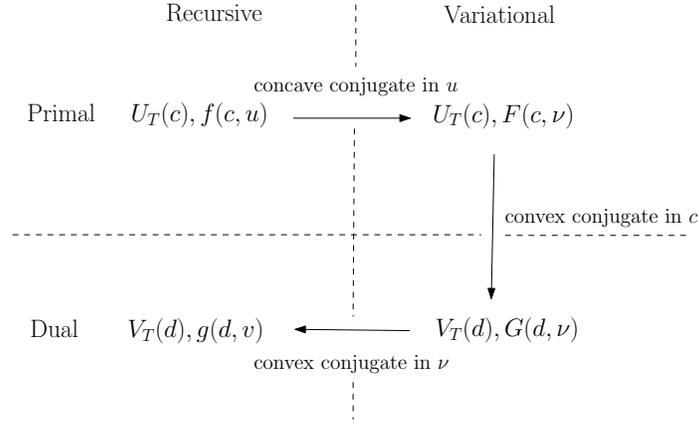}
\caption{Double Fenchel-Legendre transformation}
\end{figure}

\begin{rem}\label{rem:f-concave}
 When $f(c,u)$ is concave in $u$, we can replace the supremum (resp. infimum) in \eqref{eq:F}, \eqref{eq:dual-fF}, and \eqref{eq:g} by infimum (resp. supremum). Let Assumption \ref{ass:homothetic} holds where $\tfrac{F}{1-\gamma}$ is convex. On the other hand, since one can show $G(d, \nu)$ is convex in $\nu$, Assumption \ref{ass:G-concave} is no longer needed. Then the same statement of Proposition \ref{prop:duality-bd} holds when $\nu^{yD}:=-g_v(D, V^{yD}) \in \cV$ for any $y>0$ and $D\in \cD_a$.
\end{rem}

\subsection{Epstein-Zin preference}
The general setting described in the previous section will be specified to stochastic differential utility of \emph{Kreps-Porteus} or \emph{Epstein-Zin} type in this section. To describe this preference, let $\delta>0$ represent the discounting rate, $0<\gamma \neq 1$ be the relative risk aversion, and $0<\psi \neq 1$ be the \define{elasticity of intertemporal substitution} (EIS). Define the Epstein-Zin aggregator $f$ via
\begin{equation}\label{eq:EZ agg}
 f(c, u):= \delta \frac{c^{1-\tfrac{1}{\psi}}}{1-\tfrac{1}{\psi}} ((1-\gamma) u)^{1-\tfrac{1}{\theta}} - \delta \theta u, \quad \text{ for } c>0 \text{ and } (1-\gamma) u>0,
\end{equation}
where $\theta:= \tfrac{1-\gamma}{1-1/\psi}$.
We consider bequest utility $U_T(c) = \tfrac{c^{1-\gamma}}{1-\gamma}$ as in Assumption \ref{ass:homothetic}.

Direct calculations specify various functions defined in the previous section.
\begin{lem}\label{lem:EZ-fun}
 The following statements hold:
 \begin{enumerate}
  \item[(i)] $f(c,u)$ is concave in $c$, and convex in $u$ if and only if $\gamma \psi>1$.
  \item[(ii)]
  \[
   - f_u(c, u) = \delta (1-\theta) c^{1- \tfrac{1}{\psi}} ((1-\gamma)u)^{-\tfrac{1}{\theta}} + \delta \theta.
  \]
  \item[(iii)] When $\tfrac{\delta\theta - \nu}{\theta -1}>0$,
  \begin{align*}
   F(c, \nu) = \delta^\theta \frac{c^{1-\gamma}}{1-\gamma} \pare{\frac{\delta \theta - \nu}{\theta -1}}^{1-\theta}, \quad
   F_{\nu\nu}(c, \nu) = \delta^\theta \frac{\psi}{1-\gamma \psi} c^{1-\gamma} \pare{\frac{\delta \theta - \nu}{\theta -1}}^{-1-\theta}.
  \end{align*}
  Therefore Assumption \ref{ass:homothetic} holds if and only if $\gamma \psi>1$.
  \item[(iv)] When $\tfrac{\delta\theta - \nu}{\theta -1}>0$,
  \[
   G(d, \nu) = \delta^{\tfrac{\theta}{\gamma}} \frac{\gamma}{1-\gamma} d^{\tfrac{\gamma-1}{\gamma}} \pare{\frac{\delta \theta-\nu}{\theta-1}}^{\tfrac{1-\theta}{\gamma}}, \quad G_{\nu\nu}(d, \nu) = \delta^{\tfrac{\theta}{\gamma}} \frac{1}{\gamma(1-\gamma \psi)} d^{\tfrac{\gamma-1}{\gamma}} \pare{\frac{\delta \theta - \nu}{\theta -1}}^{\tfrac{1-\theta}{\gamma}-2}.
  \]
  Therefore Assumption \ref{ass:G-concave} holds if and only if $\gamma \psi>1$.
  \item[(v)] For $(1-\gamma)v >0$,
  \begin{align*}
   g(d, v) = \delta^\psi \frac{d^{1-\psi}}{\psi-1} ((1-\gamma) v)^{1-\tfrac{\gamma\psi}{\theta}} - \delta \theta v, \quad
   -g_v(d, v) &= \delta^\psi (1-\theta) d^{1-\psi} ((1-\gamma) v)^{-\tfrac{\gamma\psi}{\theta}} + \delta \theta.
  \end{align*}
 \end{enumerate}
\end{lem}

Let us now recall several sufficient conditions for the existence of Epstein-Zin utility.
\begin{prop}\label{prop:EZ-existence}
 Let the filtration $(\F_t)_{0\leq t\leq T}$ be the augmented filtration generated by some Brownian motion.
 \begin{enumerate}
  \item[(i)] \cite[Theorem 1]{Schroder-Skiadas-JET} When either $\gamma>1, 0<\psi<1$, or $0<\gamma<1, \psi>1$, for any $c\in \C$ such that $\expec[\int_0^T c^\ell_t dt + c_T^\ell]<\infty$ for all $\ell\in \Real$, there exists a unique $U^c$ such that $\expec[\esssup_t |U^c_t|^\ell] <\infty$ for every $\ell>0$. Moreover $U^c_0$ is concave in $c$.
  \item[(ii)] \cite[Propositions 2.2 and 2.4]{Xing} When $\gamma, \psi>1$, for any $c\in \C$ such that $\expec[\int_0^T c_t^{1-1/\psi}dt +  c_T^{1-\gamma}] <\infty$, there exists a unique $U^c$ of class (D). Moreover $U^c_0$ is concave in $c$.
 \end{enumerate}
\end{prop}
\begin{rem}
 When $\theta<1$, \cite{Duffie-Lions} shows the existence of $U^c$ in a Markovian setting. When the assumption on filtration in Proposition \ref{prop:EZ-existence} is removed, \cite[Theorems 3.1 and 3,3]{Seiferling-Seifried} proves the statement of part (ii) for $c\in \C$ such that $\expec[\int_0^T c^\ell_t dt + c_T^\ell]<\infty$ for all $\ell\in \Real$.
\end{rem}

The previous result indicates that, for different values of $\gamma$ and $\psi$, Epstein-Zin utility exists and is of class (D). Therefore we set the class of process $\cU$ as
\[
 \cU := \{U\,|\, \text{progressively measurable}, (1-\gamma) U\geq 0, \text{ and is of class (D)}\}.
\]

On the dual side, the following result provides a sufficient condition on the existence of stochastic differential dual, implying $\cD_a \neq \emptyset$.
\begin{prop}\label{prop:sdd-existence}
 Let the filtration $(\F_t)_{0\leq t\leq T}$ be the augmented filtration generated by some Brownian motion. Consider the following equation for $V^{yD}$:
  \begin{equation}\label{eq:sdd-eq}
  V^{yD}_t = \expec_t\bra{\tfrac{\gamma}{1-\gamma} (yD_T)^{\tfrac{\gamma-1}{\gamma}} + \int_t^T \tfrac{\delta^\psi}{\psi-1} (yD_s)^{1-\psi} \pare{\tfrac{1-\gamma}{\gamma} V^{yD}_s}^{1-\tfrac{\gamma \psi}{\theta}} - \tfrac{\delta \theta}{\gamma} V^{yD}_s \,ds}, \quad \text{ for all } t\leq T.
 \end{equation}
 \begin{enumerate}
 \item[(i)] When either $\gamma>1, 0<\psi<1$, or $0<\gamma<1, \psi>1$, for any $y>0$ and $D\in \cD$ such that $\expec[\int_0^T D_t^{\ell} dt + D_T^\ell]<\infty$ for all $\ell\in \Real$, there exists a unique $V^{yD}$ satisfying $(1-\gamma) V^{y D}\geq 0$, \eqref{eq:sdd-eq}, and $\expec[\esssup_t |V^{yD}_t|^\ell]<\infty$ for every $\ell>0$.
 \item[(ii)] When $\gamma, \psi>1$, for any $y>0$, $D\in \cD$ such that $\expec[\int_0^T D_t^{1-\psi} dt + D_T^{(\gamma-1)/\gamma}]<\infty$, there exists a unique $V^{yD}$ of class (D) satisfying $(1-\gamma) V^{yD}\geq 0$ and \eqref{eq:sdd-eq}.
 \end{enumerate}

\end{prop}

For variational representations, we choose
\[
 \cV:= \{\nu \,|\, \text{progressively measurable and } \nu \geq \delta \theta\}.
\]
The choice of $\cV$ implies that that $\kappa^\nu_{s,t}$ is bounded for any $\nu\in \cV$ and $0\leq s\leq t\leq T$. Now we are ready to report the main result of this section.

\begin{thm}\label{thm:EZ-dual}
 Consider the Epstein-Zin utility whose aggregator $f(c,u)$ is convex in $u$, i.e., $\gamma\psi>1$. Then the inequality \eqref{eq:dual-ineq} holds under following parameter specification:
 \begin{enumerate}
  \item[(i)] $0<\gamma<1, \gamma \psi>1$;
  \item[(ii)] $\gamma, \psi>1$.
 \end{enumerate}
\end{thm}

\begin{rem}
 When $\gamma=1/\psi$, Epstein-Zin utility reduces to time-additive utility with constant relative risk aversion $\gamma$. Then \eqref{eq:sdd-eq} reduces to the following standard form of the dual problem
 \[
  V^{yD}_t = \expec_t \bra{\tfrac{\gamma}{1-\gamma} e^{-\tfrac{\delta}{\gamma} T} (y D_T)^{\tfrac{\gamma-1}{\gamma}} + \int_t^T \delta^{\tfrac{1}{\gamma}} \tfrac{\gamma}{1-\gamma} e^{-\tfrac{\delta}{\gamma} s} (y D_s)^{\tfrac{\gamma-1}{\gamma}} ds}.
 \]
\end{rem}

\section{Main results}\label{sec: no gap}
\subsection{Candidate optimal strategies}
For Epstein-Zin utility and a wide class of financial models, we will show that the inequality \eqref{eq:dual-ineq} is actually an identity, i.e., there is \emph{no duality gap}. Moreover we will identify $(\pi^*, c^*)$ and $(y^*, D^*)$ such that
\begin{equation}\label{eq:no gap}
 \max_{(\pi, c)\in \cA} U^c_0 = U^{c^*}_0 = V_0^{y^* D^*} + w y^* = \min_{y>0}(\min_{D\in \cD_a} V_0^{yD} + w y).
\end{equation}
Therefore, $(\pi^*, c^*)$ (resp. $D^*$) is the optimizer for the primal (resp. dual) problem, and $y^*$ is the Lagrangian multiplier.

We will work with models with Brownian noise. Let $(\mathcal{F}_t)_{0\leq t\leq T}$ be the argumented filtration generated by a $k+n$-dimensional Brownian motion $B = (W, W^\bot)$, where $W$ (resp. $W^\bot$) represents the first $k$ (resp. last $n$) components. We will also use $(\F^W_t)_{0\leq t\leq T}$ (resp. $(\F^{W^\bot}_t)_{0\leq t\leq T}$) as the argumented filtration generated by $W$ (resp. $W^\bot$).
Consider a model of financial market where assets $S= (S^0, S^1, \dots, S^n)$ have the dynamics
\begin{equation}\label{eq:SDE S}
\begin{split}
 dS^0_t = S^0_t r_t dt, \quad
 dS^i_t = S^i_t \Big[(r_t + \mu^i_t) dt + \sum_{j=1}^n \sigma^{ij}_t dW^{\rho, j}_t\Big], \quad i=1, \dots, n.
\end{split}
\end{equation}
Here $r, \mu, \sigma$ and $\rho$ are $\F^W$-adapted processes valued in $\Real, \Real^n, \Real^{n\times n}, \Real^{n\times k}$, respectively, and satisfy $\int_0^T |\alpha_t|^2 dt<\infty$ a.s. for $\alpha = r, \mu, \sigma, \rho$, and $\sigma \sigma'$ is assumed to be invertible. The $n$-dimensional Brownian motion $W^\rho$ is defined as $W^\rho:= \int_0^\cdot \rho_s dW_s + \int_0^\cdot \rho^\bot_s dW^\bot_s$ for a $\Real^{n\times n}$-valued process $\rho^\bot$ satisfying $\rho \rho' + \rho^\bot (\rho^\bot)' = 1_{n\times n}$ (the $n$-dimensional identity matrix). Then $W^\rho$ and $W$ has (instantaneous) correlation $\rho$. For $(\pi, c)\in \mathcal{A}$, $\mathcal{W}^{(\pi, c)}$ follows
\begin{equation}\label{eq:SDE W}
 d\mathcal{W}^{(\pi, c)}_t = \mathcal{W}^{(\pi, c)}_t [(r_t + \pi'_t \mu_t) dt + \pi'_t \sigma_t dW^\rho_t] - c_t dt.
\end{equation}

Consider the \define{primal} and \define{dual value processes} defined as
\[
 \mathbb{U}^{c}_t := \esssup_{(\tilde{\pi}, \tilde{c})\in \cA(\pi, c, t)} U^{\tilde{c}}_t \quad \text{ and } \quad  \mathbb{V}^{yD}_t :=\essinf_{\tilde{D}\in \cD_a(D, t)} V^{y\tilde{D}}_t,
\]
where
\[
 \cA(\pi, c, t) := \{(\tilde{\pi}, \tilde{c})\in \cA \,:\, (\tilde{\pi}, \tilde{c}) = (\pi, c) \text{ on } [0,t]\}, \quad \cD_a(D, t) := \{\tilde{D}\in \cD_a \,:\, \tilde{D} = D \text{ on } [0,t]\}.
\]
Due to the homothetic property of Epstein-Zin utility, we speculate that $\mathbb{U}^c$ and $\mathbb{V}^{yD}$ have the following decomposition:
\begin{equation}\label{eq:value-proc}
 \mathbb{U}^{c}_t = \tfrac{1}{1-\gamma} (\cW^{(\pi, c)}_t)^{1-\gamma} e^{Y^p_t} \quad \text{ and } \quad \mathbb{V}^{yD}_t = \tfrac{\gamma}{1-\gamma} (y D_t)^{\frac{\gamma-1}{\gamma}} e^{Y^d_t/\gamma},
\end{equation}
for some processes $Y^p$ and $Y^d$. Let us derive the dynamic equations that $Y^p$ and $Y^d$ satisfy via the \emph{martingale principle}: $\bU^c + \int_0^\cdot f(c_s, \bU^c_s) ds$ (resp. $\bV^{y D} + \int_0^\cdot g(yD_s, \tfrac{1}{\gamma} \bV^{yD}_s) ds$) is a supermartingale (resp. submartingale) for arbitrary $(\pi, c)$ (resp. $D$) and is a martingale for the optimal one. For Markovian models, the martingale principle is a reformulation of the \emph{dynamic programming principle}. For the non-Markovian models, it can be considered as the dynamic programming for BSDEs, cf., eg. \cite{Hu-Imkeller-Muller}.

\begin{lem}\label{lem:mart-prin}
 The ansatz \eqref{eq:value-proc} and the martingale principle imply that both $(Y^p, Z^p)$ and $(Y^d, Z^d)$, for some $Z^p$ and $Z^d$, satisfies the BSDE
 \begin{equation}\label{eq:BSDE Y}
  Y_t = \int_t^T H(Y_s, Z_s) ds -\int_t^T Z_s dW_s, \quad t\in [0,T],
 \end{equation}
 where $H: \Omega \times \Real^n \times \Real^{n\times n} \rightarrow \Real$ is given by
 \begin{equation}\label{eq:H}
  H(y,z) := \tfrac12 z M_t z' + \tfrac{1-\gamma}{\gamma} \mu_t' \Sigma^{-1}_t \sigma_t \rho_t z' + \theta \tfrac{\delta^\psi}{\psi} e^{-\tfrac{\psi}{\theta} y} + h_t - \delta \theta.
 \end{equation}
 Here, suppressing the subscript $t$,
 \begin{equation}\label{eq:func def}
  \Sigma:= \sigma \sigma', \quad M:= 1_{k\times k} + \tfrac{1-\gamma}{\gamma} \rho' \sigma' \Sigma^{-1} \sigma \rho, \quad \text{and} \quad h:= (1-\gamma) r + \tfrac{1-\gamma}{2\gamma} \mu' \Sigma^{-1} \mu.
 \end{equation}
 The function $H$, interpreted as the Hamilton of the primal and dual optimization problem, has the following representation:
 \begin{equation}\label{eq:H-rep}
 \begin{split}
  H(y,z) =& (1-\gamma) r_t - \delta \theta + \tfrac12 |z|^2 + (1-\gamma) \sup_{\overline{c}} \Big[-\overline{c} + \delta e^{-\tfrac{1}{\theta} y} \tfrac{1}{1-\tfrac{1}{\psi}} \overline{c}^{1-\tfrac{1}{\psi}}\Big]\\
  &+ (1-\gamma) \sup_{\pi} \Big[-\tfrac{\gamma}{2} \pi' \Sigma_t \pi + \pi' (\mu_t + \sigma_t \rho_t z')\Big]\\
  =& (1-\gamma) r_t -\delta \theta + \theta \tfrac{\delta^\psi}{\psi} e^{-\tfrac{\psi}{\theta} y} + \tfrac{1}{2\gamma} |z|^2\\
  &+ (1-\gamma) \inf_{\mu_t + \sigma_t \rho_t \xi' + \sigma_t \rho^\bot_t \eta'=0} \Big[\tfrac{1}{2\gamma} (|\xi|^2 + |\eta|^2) - \tfrac{1}{\gamma} \xi z'\Big].
 \end{split}
 \end{equation}
 Their optimizers, evaluated at $(Y^p, Z^p)$ for the primal problem and $(Y^d, Z^d)$ for the dual problem, are
 \begin{equation}\label{eq:optimizer}
 \begin{split}
  &\pi^*_t = \tfrac{1}{\gamma} \Sigma^{-1}_t (\mu_t + \sigma_t \rho_t (Z^p_t)'), \quad \tfrac{c_t^*}{\mathcal{W}_t^{\pi^*, c^*}} = \overline{c}^*_t= \delta^\psi e^{-\tfrac{\psi}{\theta} Y^p_t},\\
  &\xi^*_t = -( \mu_t' + Z^d \rho'_t \sigma'_t) \Sigma^{-1}_t \sigma_t \rho_t + Z^d_t, \quad
  \eta^*_t = -(\mu_t' + Z^d \rho'_t \sigma'_t) \Sigma^{-1}_t \sigma_t \rho_t^\bot,\\
  & dD^*_t / D^*_t = -r_t dt + (-(\mu_t'+ Z^d \rho_t' \sigma_t') \Sigma_t^{-1} \sigma_t dW^\rho_t + Z^d_t dW_t)= - r_t dt + (\xi^*_t dW_t + \eta^*_t dW^\bot_t).
 \end{split}
 \end{equation}
\end{lem}

In what follows, we will make the previous heuristic argument rigorous by starting from the BSDE \eqref{eq:BSDE Y} and showing that it admits a solution $(Y,Z)$. Replacing $(Y^p, Z^p)$ and $(Y^d, Z^d)$ in \eqref{eq:optimizer} by $(Y,Z)$, we call the resulting processes $(\pi^*, c^*)$ and $D^*$ the candidate optimal strategies for the primal and dual problem, respectively.
The candidate optimal strategy for the primal problem has been documented in various settings, cf. \cite[Theorem 2 and 4]{Schroder-Skiadas-JET} for complete markets, \cite[Equation (4.4)]{KSS-FS}, \cite[Theorem 6.1]{Kraft-Seiferling-Seifried}, and \cite[Equation (2.14)]{Xing} for Markovian models. The form for $D^*$ can be obtained via the utility gradient approach, cf. \cite[Equation (35)]{Duffie-Epstein-pricing}, \cite[Theorem 2.2]{Duffie-Skiadas}, and \cite[Equation (4)]{Schroder-Skiadas-JET}; see also Corollary \ref{cor:utility graident} below. The novelty here is to relate $D^*$ and the minimization problem in \eqref{eq:H-rep}. To understand this minimization problem, we start with the following class of state price densities:
\begin{equation}\label{eq:SDE D}
 dD_t/ D_t = -r_t dt + \xi_t dW_t + \eta_t dW^\bot_t, \quad \text{for some } \xi, \eta.
\end{equation}
This form ensures that $D \mathcal{W}^{(0,0)}$, where $\mathcal{W}^{(0,0)}$ is the wealth process of no investment and consumption, is a supermartingale. In general, $D \mathcal{W}^{(\pi, c)}$ satisfies
\[
 d D_t \cW^{(\pi, c)}_t = D_t \cW^{(\pi, c)}_t \pi'_t (\mu_t + \sigma_t \rho_t \xi_t' + \sigma_t \rho^\bot_t \eta_t') dt - D_t c_t dt + \text{ local martingale}.
\]
Therefore, $D\in \cD$ necessarily implies that $\mu + \sigma \rho \xi' + \sigma \rho^\bot \eta' =0$, which is the constraint for the minimization problem in \eqref{eq:H-rep}. On the other hand, calculation shows
\begin{equation}\label{eq:Ito dual}
\begin{split}
 d D_t^{\tfrac{\gamma-1}{\gamma}} &= \tfrac{\gamma-1}{\gamma} D_t^{\tfrac{\gamma-1}{\gamma}} \big[-r_t - \tfrac{1}{2r} (|\xi_t|^2 + |\eta_t|^2)\big] dt+ \tfrac{\gamma-1}{\gamma} D_t^{\tfrac{\gamma-1}{\gamma}} (\xi_t dW_t + \eta_t dW^\bot_t),\\
 de^{Y^d_t/\gamma} &= e^{Y^d_t/\gamma} \big[-\tfrac{1}{\gamma} H(Y^d_t, Z^d_t) + \tfrac{1}{2\gamma^2} |Z^d_t|^2\big] dt + e^{Y^d_t/\gamma} \tfrac{Z^d_t}{\gamma} dW_t.
\end{split}
\end{equation}
Therefore the drift of $\tfrac{\gamma}{1-\gamma} (yD)^{\tfrac{\gamma-1}{\gamma}} e^{Y^d/\gamma} + \int_0^\cdot g(yD_s, \tfrac{1}{1-\gamma} (yD_s)^{\tfrac{\gamma-1}{\gamma}} e^{Y^d_s/\gamma}) ds$ reads (after suppressing the subscript $t$)
\[
 \tfrac{1}{1-\gamma} (yD)^{\tfrac{\gamma-1}{\gamma}} e^{Y^d/\gamma} \Big\{(1-\gamma) r -\delta \theta + \theta \tfrac{\delta^\psi}{\psi} e^{-\tfrac{\psi}{\theta} Y^d} + \tfrac{1}{2\gamma} |Z^d|^2 + (1-\gamma) \big[\tfrac{1}{2\gamma}(|\xi|^2 + |\eta|^2) - \tfrac{1}{\gamma} \xi (Z^d)'\big]-H(Y^d, Z^d)\Big\}.
\]
Then the martingale principle implies that the previous drift is nonnegative, leading to the minimization problem in \eqref{eq:H-rep}. Solving this constrained minimization problem via the Lagrangian multiplier method, we obtain its minimizer in \eqref{eq:optimizer}.

\subsection{Models with bounded market price of risk}
We will verify in this section the identity \eqref{eq:no gap}, hence confirm the optimality of $(\pi^*, c^*)$ and $D^*$. To avoid technicality clouds the idea of proofs, we start from the following restriction on model coefficients.
\begin{ass}\label{ass:bdd-coeff}
 The processes $r$ and $\mu' \Sigma^{-1} \mu$ are both bounded.
\end{ass}
This assumption allows non-Markovian models, but requires the \emph{market price of risk} $\sqrt{\mu' \Sigma^{-1}\mu}$ to be bounded. Markovian models with unbounded market price of risk will be discussed in the next section, where more technical conditions will be imposed. We will also assume the same restriction on utility parameters $\gamma$ and $\psi$ as in Theorem \ref{thm:EZ-dual}.

\begin{lem}\label{lem:BSDE Y bdd}
 Suppose that either $0<\gamma<1$, $\gamma \psi >1$, or $\gamma, \psi>1$, and that Assumption \ref{ass:bdd-coeff} holds. Then \eqref{eq:BSDE Y} admit a solution $(Y,Z)$ such that $Y$ is bounded and $Z\in H_{\BMO}$\footnote{$Z\in H_{\BMO}$ if $\sup_{\tau} \|\expec_\tau[\int_\tau^T |Z_s|^2 ds]\|_{\mathbb{L}^\infty}<\infty$, where $\tau$ is chosen from the set of $\F$-stopping times}.
\end{lem}
Having establish a solution $(Y, Z)$ to \eqref{eq:BSDE Y}, we define
\begin{equation}\label{eq:opt str}
\begin{split}
 \pi^*_t = \tfrac{1}{\gamma} \Sigma^{-1}_t (\mu_t + \sigma_t \rho_t Z_t'), & \qquad \tfrac{c^*_t}{\mathcal{W}_t^{(\pi^*, c^*)}} = \delta^\psi e^{-\tfrac{\psi}{\theta} Y_t},\\
 dD^*_t/ D^*_t = -r_t dt + (-\gamma (\pi^*_t)' \sigma_t dW^\rho_t + ZdW_t), & \qquad y^* = w^{-\gamma} e^{Y_0},
\end{split}
\end{equation}
and present the main result of this section.

\begin{thm}\label{thm:opt str bdd}
Suppose that either $0<\gamma<1$, $\gamma \psi >1$, or $\gamma, \psi>1$, and that Assumption \ref{ass:bdd-coeff} holds. Then, for $\pi^*, c^*, D^*$, and $y^*$ defined in \eqref{eq:opt str},
\begin{equation}\label{eq:dual-eq}
 \max_{(\pi, c)\in \cA} U^c_0 = U^{c^*}_0 = V_0^{y^* D^*} + w y^* = \min_{y>0}(\min_{D\in \cD_a} V_0^{yD} + w y).
\end{equation}
Therefore $(\pi^*, c^*)$ is the optimal strategy for the primal problem, $D^*$ is the optimal state price density for the dual problem, and $y^*$ is the Lagrangian multiplier.
\end{thm}

As a direct consequence of Theorem \ref{thm:opt str bdd}, the minimizer $D^*$ of the dual problem is identified as the super-differential of the primal value function, coming from the utility gradient approach, cf. \cite{Duffie-Epstein-pricing}, \cite{Duffie-Skiadas}.

\begin{cor}\label{cor:utility graident}
 The state price density $D^*$ satisfies
 \begin{equation}\label{eq:utility gradient}
  D^*_t = w^\gamma e^{-Y_0} \exp\Big[\int_0^t \partial_u f(c^*_s, U^{c^*}_s) ds\Big] \partial_c f(c^*_t, U^{c^*}_t), \quad t\in[0,T].
 \end{equation}
 Moreover, when assumptions of Theorem \ref{thm:opt str bdd} hold, $\mathcal{W}^{(\pi^*, c^*)} D^* + \int_0^\cdot D^*_s c^*_s ds$ is a martingale.
\end{cor}

\subsection{Models with unbounded market price of risk}
Many widely used market models in the asset pricing literature come with unbounded market price of risk; for example, Heston model in \cite{Chacko-Viceira}, \cite{Kraft}, and \cite{Liu}, Kim-Omberg model in \cite{Kim-Omberg} and \cite{Wachter-port}. To obtain similar result as Theorem \ref{thm:opt str bdd} and Corollary \ref{cor:utility graident}, we focus on the utility specification
$\gamma, \psi>1$,
and work with Markovian models, whose investment opportunities are driven by a state variable $X$ satisfying
\begin{equation}\label{eq:SDE X}
 dX_t = b(X_t) dt + a(X_t) dW_t.
\end{equation}
Here $X$ takes value in an open domain $E \subseteq \Real^k$, $b: E\rightarrow \Real^k$ and $a: E\rightarrow \Real^{k\times k}$. Given functions $r: E \rightarrow \Real$, $\mu: E \rightarrow \Real^n, \sigma: E\rightarrow \Real^{n\times n}$, and $\rho: E \rightarrow \Real^{n\times k}$, the processes $r, \mu, \sigma, \rho$ in \eqref{eq:SDE S} are corresponding functions evaluated at $X$. Instead of Assumption \ref{ass:bdd-coeff}, these model coefficients satisfy the following assumptions.

\begin{ass}\label{ass:gen-coeff}
 $r, \mu, \sigma, b, a$, and $\rho$ are all locally Lipschitz in $E$; $A:= aa'$ and $\Sigma= \sigma \sigma'$ are positive definite in any compact subdomain of $E$; dynamics of \eqref{eq:SDE X} does not reach the boundary of $E$ in finite time; moreover $r + \tfrac{1}{2\gamma} \mu' \Sigma^{-1} \mu$ is bounded from below on $E$.
\end{ass}
The regularity of coefficients and the nonexplosion assumption ensure that the dynamics for $X$ is wellposed, i.e., \eqref{eq:SDE X} admits a unique $E$-valued strong solution $(X_t)_{0\leq t\leq T}$. The assumption on the lower bound of  $r + \tfrac{1}{2\gamma} \mu' \Sigma^{-1} \mu$ allows for unbounded market price of risk and is readily satisfied when $r$ is bounded from below.

To present analogue of Theorem \ref{thm:opt str bdd} and Corollary \ref{cor:utility graident}, let us first introduce two sets of abstract conditions, which will be verified in two classes of models below.

\begin{ass}\label{ass:over P}
$\,$
\begin{itemize}
 \item[(i)] $\tfrac{d\overline{\prob}}{d\prob}= \mathcal{E}\big(\int \tfrac{1-\gamma}{\gamma} \mu' \Sigma^{-1} \sigma \rho(X_s) dW_s\big)$ defines a probability measure $\overline{\prob}$ equivalent to $\prob$;
 \item[(ii)] $\expec^{\overline{\prob}}\big[\int_0^T h(X_s) ds\big]>-\infty$, where $h$ comes from \eqref{eq:func def}.
\end{itemize}
\end{ass}
When all model coefficients are bounded, as in Assumption \ref{ass:bdd-coeff}, Assumption \ref{ass:over P} is automatically satisfied. When the market price of risk is unbounded, the last part of Assumption \ref{ass:gen-coeff} and $\gamma>1$ combined imply that $h$ is bounded from above by $h_{\max} := \max_{x\in E} h(x)$, but is not bounded from below. Nevertheless Assumption \ref{ass:over P} allows us to transform \eqref{eq:BSDE Y} under $\overline{\prob}$ and present the following result from \cite[Proposition 2.9]{Xing}.

\begin{lem}\label{lem:BSDE Y}
 Let Assumptions \ref{ass:gen-coeff} and \ref{ass:over P} hold. For $\gamma, \psi>1$, \eqref{eq:BSDE Y} admits a solution $(Y, Z)$ such that, for any $t\in [0,T]$,
  \begin{multline}\label{eq: Y bdd}
         \expec^{\overline{\prob}}_t \Big[\int_t^T h(X_s) \, ds\Big] - \delta \theta(T-t) +\theta \frac{\delta^\psi}{\psi} e^{(\delta \psi-\frac{\psi}{\theta} h_{\max}) T} (T-t) \leq Y_t \leq\\
          \leq -\delta \theta (T-t)+\log\expec^{\overline{\prob}}_t\Big[ \exp{\Big(\int_t^T h(X_s)\,ds\Big)}\Big],
        \end{multline}
 and $\expec^{\overline{\prob}}[\int_0^T |Z_s|^2 ds] <\infty$.
 In particular, since $h\leq h_{max}$, $Y$ is bounded from above.
\end{lem}

Having constructed $(Y,Z)$, $(\pi^*, c^*)$ and $D^*$ in \eqref{eq:opt str} are well defined. To verify their optimality, let us introduce an operator $\fF$. For $\phi\in C^2(E)$,
\begin{equation}\label{eq: fF}
 \fF[\phi]:= \tfrac12 \sum_{i,j=1}^k A_{ij} \partial^2_{x_i x_j} \phi + \Big(b+ \tfrac{1-\gamma}{\gamma} a \rho' \sigma' \Sigma^{-1} \mu\Big)' \nabla \phi + \tfrac12 \nabla \phi' a M a' \nabla \phi + h,
\end{equation}
where the dependence on $x$ is suppressed on both sides. To understand this operator, note that the solution $(Y,Z)$ to \eqref{eq:BSDE Y} is expected to be Markovian, i.e., there exists a function $u: [0,T]\times E \rightarrow \Real$ such that $Y= u(\cdot, X)$. Then the BSDE \eqref{eq:BSDE Y} corresponds the following PDE:
\[
 \partial_t u + \fF[u] + \theta \tfrac{\delta^\psi}{\psi} e^{-\tfrac{\psi}{\theta}u} -\delta \theta=0, \quad u(T, x) = 0.
\]
Since $\theta<0$ when $\gamma, \psi>1$, moreover $Y$, hence $u$, is bounded from above, therefore the last two terms in the previous PDE are bounded, then $\fF$ is the unbounded part of the spatial operator.

\begin{ass}\label{ass: phi}
 There exists $\phi\in C^2(E)$ such that
 \begin{enumerate}
  \item[(i)] $\lim_{n\rightarrow \infty} \inf_{x\in E\setminus E_n} \phi(x) = \infty$, where $(E_n)_n$ is a sequence of open domains in $E$ satisfying $\cup_{n} E_n = E$, $\overline{E}_n$ compact, and $\overline{E}_n \subset E_{n+1}$, for each $n$;
  \item[(ii)] $\fF[\phi]$ is bounded from above on $E$.
 \end{enumerate}
\end{ass}

The function $\phi$ in the previous assumption is called a \emph{Lyapunov} function. Its existence facilities to prove that certain exponential local martingale is in fact a martingale, leading to the following result.

\begin{thm}\label{thm:opt str unbdd}
 Suppose that $\gamma, \psi>1$, and that Assumptions \ref{ass:gen-coeff}, \ref{ass:over P}, \ref{ass: phi} hold. Then the statements of Theorem \ref{thm:opt str bdd} and Corollary \ref{cor:utility graident} hold.
\end{thm}

\begin{rem}\label{rem:relax cond}
 The optimality of $(\pi^*, c^*)$ has been verified in \cite[Theorem 2.14]{Xing} under more restrictive conditions. First, \cite{Xing} restricts strategies to a \emph{permissible} class which is smaller than the current admissible class $\cA$. It is the duality inequality \eqref{eq:dual-ineq} that allows us to make this extension.
 Second \cite[Assumption 2.11]{Xing} is needed to ensure $c^*$ satisfying the integrability condition in Proposition \ref{prop:EZ-existence} (ii). This integrability condition translates to model parameter restrictions, see \cite[Proposition 3.2 ii)]{Xing} for Heston model and \cite[Proposition 3.4 ii)]{Xing} for Kim-Omberg model. Rather than forcing $c^*$ to satisfy this integrability condition, which is a sufficient condition for the existence of Epstein-Zin utility, we show that Epstein-Zin utility exists for $c^*$, hence $c^*$ belongs to $\C_a$, which abstractly envelops all Epstein-Zin utilities and, in particular, contains those ones satisfying the integrability condition. As a result the aforementioned model parameter restrictions for Heston model and Kim-Omberg model can be removed.
\end{rem}

\begin{exa}[Stochastic volatility]
 Consider a $1$-dimensional process $X$ following
 \[
  dX_t = b(\ell- X_t) dt + a\sqrt{X_t} dW_t,
 \]
 where $b, \ell\geq 0, a>0$, and $b\ell >\tfrac12 a^2$.
 Given $r_0, r_1\in \Real$, $\sigma: (0,\infty)\rightarrow \Real^{n\times n}$ and $\lambda: (0,\infty)\rightarrow \Real^n$, which are locally Lipschitz continuous on $(0,\infty)$ and $\Sigma(x):= \sigma \sigma(x)'>0$,
 let $r(X) = r_0 + r_1 X$ be the interest rate, $\sigma(X)$ be the volatility of risky assets,  $\mu(X)= \sigma(X) \lambda(X)$ be the excess return, and the dynamics of assets follow \eqref{eq:SDE S} with $\rho \in \Real^n$. This class of models encapsulate 1) Heston model studied in \cite{Kraft} and \cite{Liu} where $n=1$, $\lambda(x)=\lambda \sqrt{x}$ for a $\lambda \in \Real$ and $\sigma(x) = \sqrt{x}$, and 2) an inverse Heston model studied in \cite{Chacko-Viceira} where $n=1$, $\lambda(x) = \lambda \sqrt{x}$ for a $\lambda \in \Real$ and $\sigma(x) = \tfrac{1}{\sqrt{x}}$. Set $\Theta(x):= \sigma(x)' \Sigma(x)^{-1} \sigma(x)$.
 The following result specifies Assumptions \ref{ass:gen-coeff}, \ref{ass:over P}, and \ref{ass: phi} to explicit model parameter restriction.
 \begin{prop}\label{prop:heston}
  Assume that $\lambda(x)= \lambda \sqrt{x}$, for some $\lambda\in \Real^n$, and $r_1 + \frac{1}{2\gamma} \lambda' \Theta(x) \lambda\geq 0$. Then for $\gamma, \psi>1$ the statements of Theorem \ref{thm:opt str bdd} and Corollary \ref{cor:utility graident} hold when either $r_1>0$ or $\lambda' \Theta(x) \lambda > 0$.
 \end{prop}
\end{exa}

\begin{exa}[Linear diffusion]
 Consider a $1$-dimensional Ornstein-Uhlenbeck process $X$ following
 \[
  dX_t = -bX_t dt + a dW_t,
 \]
where $a,b>0$. Given $\lambda_0, \lambda_1 \in \Real^n$ and $\sigma\in \Real^{n\times n}$ with $\Sigma := \sigma \sigma'>0$, let $r(X) = r_0 + r_1 X$ be the interest rate, $\sigma(X) = \sigma$ be the volatility of risky assets, and $\mu(X) = \sigma(\lambda_0+ \lambda_1 X)$ be the excess return, and the dynamics of assets follow \eqref{eq:SDE S} with $\rho \in \Real^n$. This model has been studied by \cite{Kim-Omberg} and \cite{Wachter-port} for time separable utility, and by \cite{Campbell-Viceira} for recursive utility in discrete time. Set $\Theta:= \sigma' \Sigma^{-1} \sigma$.
 The following result from \cite[Proposition 3.4]{Xing} specifies Assumptions \ref{ass:gen-coeff}, \ref{ass:over P}, and \ref{ass: phi} to explicit model parameter restriction.
 \begin{prop}\label{prop:Kim-Omberg}
 Assume that either of the following parameter restrictions hold:
 \begin{enumerate}
 \item[(i)] $r_1=0$ and $-b+ \tfrac{1-\gamma}{\gamma} a \lambda'_1 \Theta \rho <0$;
 \item[(ii)] $\lambda'_1 \Theta \lambda_1>0$.
 \end{enumerate}
  Then for $\gamma, \psi>1$ the statements of Theorem \ref{thm:opt str bdd} and Corollary \ref{cor:utility graident} hold.
 \end{prop}
\end{exa}

\appendix
 \section{Proofs}
 \begin{proof}[Proof of Lemma \ref{lem:var-rep}]
  It suffices to check the statement in part (i) for $\nu$ with $U_0^{c, \nu}>-\infty$. When $U_0^{c, \nu}$ has finite value,
  \eqref{eq:sdu} and part (i) of Definition \ref{def: dualvar} imply that $U^c + \int_0^\cdot f(c_s, U^c_s) ds$ and $\kappa^\nu_{0, \cdot}U^{c, \nu} + \int_0^\cdot \kappa^\nu_{0, s} F(c_s, \nu_s)ds$ are both martingales. Then $U^{c, \nu} + \int_0^\cdot F(c_s, \nu_s) - \nu_s U^{c, \nu}_s ds$ is a local martingale by It\^{o}'s formula. Therefore there exists a local martingale $L$ such that
  \[
   d(U^c_t- U^{c, \nu}_t) - \nu_t (U^c_t - U^{c, \nu}_t) = - dA_t + dL_t,
  \]
  where $A_t = \int_0^t f(c_s, U^c_s) - (F(c_s, \nu_s) - \nu_s U^c_s) ds$ is an increasing process due to \eqref{eq:dual-fF}. As a result, $\kappa^\nu_{0, \cdot}(U^c - U^{c, \nu})$ is a local super-martingale. On the other hand, Definition \ref{def: dualvar} part (ii) and (iii), together with $U^c \in \cU$, imply that $\kappa^\nu_{0, \cdot}(U^c - U^{c, \nu})$ is of class (D), hence it is a supermartingale. Therefore
  \begin{equation}\label{eq:Uc-Ucnu}
   U^c_t - U^{c, \nu}_t \geq \expec_t \big[\kappa^\nu_{t, T} (U^c_T - U^{c, \nu}_T)\big]=0.
  \end{equation}
  Taking supremum in $\nu$, we confirm the claim in part (i). For the statement in (ii), for $\nu^c \in \cV$, we have $A\equiv 0$, hence $\kappa^{\nu^c}_{0, \cdot}(U^c - U^{c, \nu^c})$ is a local martingale, and a martingale, due to its class (D) property. As a result, the inequality in \eqref{eq:Uc-Ucnu} is an identity for $\nu= \nu^c$.
 \end{proof}

 \begin{proof}[Proof of Proposition \ref{prop:sdd-existence}]
  Let the filtration be generated by some Brownian motion $B$. Solving \eqref{eq:sdd-eq} is equivalent to solve the following BSDE
  \begin{equation}\label{eq: dual BSDE EZ}
 V^{yD}_t = \tfrac{\gamma}{1-\gamma} (yD_T)^{\frac{\gamma-1}{\gamma}} + \int_t^T \tfrac{ \delta^\psi}{\psi-1} (y D_s)^{1-\psi} \pare{\tfrac{1-\gamma}{\gamma} V^{yD}_s}^{1-\tfrac{\gamma \psi}{\theta}} - \tfrac{\delta \theta}{\gamma} V^{yD}_s \, ds - \int_t^T Z^{yD}_s dB_s.
\end{equation}
Set $Y_t = \tfrac{1-\gamma}{\gamma} e^{-\frac{\delta \theta}{\gamma} t} V^{yD}_t$ and $Z_t = \tfrac{1-\gamma}{\gamma} e^{-\frac{\delta \theta}{\gamma} t} Z^{yD}_t$. The previous BSDE translates to
\begin{equation}\label{eq: dual BSDE trans}
 Y_t = e^{-\frac{\delta \theta}{\gamma} T} (y D_T)^{\frac{\gamma-1}{\gamma}} + \int_t^T \delta^\psi \tfrac{\theta}{\gamma \psi} e^{-\delta \psi s} (y D_s)^{1-\psi} Y_s^{1-\tfrac{\gamma \psi}{\theta}} ds - \int_t^T Z_s dB_s.
\end{equation}

(i) Define $\cY= Y^{\tfrac{\gamma \psi}{\theta}}$ and $\cZ = \tfrac{\gamma \psi}{\theta} Y^{\tfrac{\gamma \psi}{\theta}-1} Z$. Then $(\cY, \cZ)$ satisfies
\[
 \cY_t = e^{-\delta \psi T} (y D_T)^{1-\psi} + \int_t^T \delta^\psi e^{-\delta \psi s} (y D_s)^{1-\psi} + \tfrac12 \pare{\tfrac{\theta}{\gamma \psi}-1} \tfrac{\cZ^2_s}{\cY_s} \, ds -\int_t^T \cZ_s dB_s.
\]
This is exactly the type of BSDE studied in \cite[Equation (A7)]{Schroder-Skiadas-JET}. It then follows from \cite[Theorem A2]{Schroder-Skiadas-JET} that the previous BSDE admits a unique solution $(\cY, \cZ)$ with $\expec[\esssup_t |\cY_t|^\ell]<\infty$ for any $\ell>0$. (To treat the terminal condition $e^{-\delta \psi T} (yD_T)^{1-\psi}$, we consider an approximated terminal condition $\epsilon + e^{-\delta \psi T} (yD_T)^{1-\psi}$ with $\epsilon>0$ and its associated solution $(\cY^\epsilon, \cZ^\epsilon)$. Proceed as the proof of \cite[Theorem A2]{Schroder-Skiadas-JET}, $\cY$ is constructed as $\lim_{\epsilon \downarrow 0} \cY^\epsilon$.) Coming back to $(Y,Z)$, the statement in (i) is confirmed.

(ii) Our assumption on $D$ implies the integrability of $e^{-\tfrac{\delta \theta}{\gamma} T}(y D_T)^{\tfrac{\gamma-1}{\gamma}}$ and $\int_0^T e^{-\delta \psi s} (y D_s)^{1-\psi} ds$. Moreover, since $\gamma, \psi>1$, we have $\theta<0$, therefore the generator of \eqref{eq: dual BSDE trans} is decreasing in the $Y$-component. This is exactly the type of BSDEs studied in \cite[Proposition 2.2]{Xing}. Then the statement in (ii) is confirmed following the proof of \cite[Proposition 2.2]{Xing}.
 \end{proof}

\begin{proof}[Proof of Theorem \ref{thm:EZ-dual}]

 (ii) When $\gamma,\psi>1$, then $\theta<0$. Therefore $\tfrac{\delta\theta-\nu}{\theta-1}>0$ for any $\nu \in \cV$. Lemma \ref{lem:EZ-fun} part (iii) and (iv) yield $F, G, V_T\leq 0$, implying that $U^{c, \nu}$ and $V^{yD,\nu}$ are both nonpositive, hence $U^{c, \nu}, V^{yD, \nu}<\infty$, for any $y>0, c\in \C_a, D\in \cD_a$, and $\nu\in \cV$. When $U_0^{c,\nu}>-\infty$, $F\leq 0$ yields
 \[
  \expec_t \Big[\kappa^\nu_{0,T} U_T(c_T) + \int_0^T \kappa^\nu_{0,s} F(c_s, \nu_s) ds\Big]\leq \kappa^\nu_{0, t} U^{c, \nu}_t \leq \expec_t [\kappa^{\nu}_{0,T} U_T(c_T)].
 \]
 implying the class (D) property of $\kappa^\nu_{0,\cdot} U^{c,\nu}$. The boundedness of $\kappa^\nu$ also implies $\kappa^\nu_{0,\cdot} U$ is of class (D) for any $U\in \cU$. Similar properties can be verified for $V^{yD, \nu}$ as well. Therefore, our choice of $\cV$ satisfies Definition \ref{def: dualvar}. On the other hand, since $(1-\gamma)U^c\geq 0$ and $\theta<0$, Lemma \ref{lem:EZ-fun} part (ii) yields $\nu^c = -f_u(c, U^c)\geq \delta \theta$, implying $\nu^c\in \cV$. Now Assumptions \ref{ass:convex}, \ref{ass:homothetic}, and \ref{ass:G-concave} are verified in Lemma \ref{lem:EZ-fun}, the statement then follows from Proposition \ref{prop:duality-bd}.

 (i) The proof in this case is more involved. When $0<\gamma<1$ and $\gamma \psi>1$, we have $0<\theta<1$.
 Therefore $\tfrac{\delta \theta - \nu}{\theta -1}>0$ for any $\nu \in \cV$. Lemma \ref{lem:EZ-fun} part (iii) and (iv) yield $F,G, V_T\geq 0$. Therefore more argument is needed to ensure the existence of $U^{c, \nu}$ and $V^{yD, \nu}$. To this end, for $c\in \C_a$, let $U^c$ be the associated stochastic differential utility, and define an increasing process $A^\nu= \int_0^\cdot f(c_s, U^c_s) - (F(c_s, \nu_s) -\nu_s U^c_s) ds$. Equation \eqref{eq:sdu} then implies that $U^c + \int_0^\cdot F(c_s, \nu_s) - \nu_s U^c_s ds + A$ is a martingale for any $\nu\in \cV$, hence It\^{o}'s formula implies that $\kappa^\nu_{0, \cdot} U^c + \int_0^\cdot \kappa^\nu_{0,s} F(c_s, \nu_s) ds$ is a local supermartingale. Taking a localization sequence $(\tau_n)_{n}$, we have
 \[
   U^c_t \geq \expec_t \bra{\kappa^\nu_{t,\tau_n\wedge T} U^c_{\tau_n\wedge T} + \int_t^{\tau_n\wedge T} \kappa^\nu_{t,s} F(c_s, \nu_s) ds} \quad \text{on } \{\tau_n \geq t\}.
 \]
 Sending $n\rightarrow \infty$ on the right-hand side, the class (D) property of $U^c$ and monotone convergence theorem implies
 \[
  \expec_t \bra{\int_t^T \kappa^\nu_{t,s} F(c_s, \nu_s) ds} <\infty, \quad \text{for any } t\in [0,T], \nu\in \cV.
 \]
 Combined with $\expec_t[\kappa^\nu_{t,T} U_T(c_T)]<\infty$, it follows $U^{c, \nu}$ in Definition \ref{def: dualvar} part (i) is well defined for any $c\in \C_a$ and $\nu\in \cV$. Similar argument applied to the dual side ensures that $V^{yD, \nu}$ is also well define for any $y>0, D\in \cD_a,$ and $\nu\in \cV$.
 The statement is then confirmed by following similar argument as in the previous case.
\end{proof}

\begin{proof}[Proof of Lemma \ref{lem:mart-prin}]
 The statement for the primal problem is proved in \cite{Xing}, see the argument leading to equation (2.14) therein. In particular, since all investment opportunities are driven by $W$, it suffices to consider the martingale part of $Y$ in \eqref{eq:BSDE Y} as a stochastic integral with respect to $W$. Let us outline the argument for the primal problem.  Parameterize $c$ by $c= \overline{c} \cW$. Calculation shows
 \begin{equation}\label{eq:Ito primal}
 \begin{split}
  d\cW_t^{1-\gamma} &= (1-\gamma)\cW^{1-\gamma}_t \big[r_t - \overline{c}_t + \pi'_t \mu_t - \tfrac{\gamma}{2} \pi'_t \Sigma_t \pi_t\big] dt + (1-\gamma) \cW^{1-\gamma}_t \pi'_t \sigma_t dW^\rho_t,\\
  de^{Y^p_t} &= e^{Y^p_t} \big[-H(Y^p_t, Z^p_t) + \tfrac12 |Z^p_t|^2\big] dt + e^{Y^p_t} Z^p_t dW_t.
 \end{split}
 \end{equation}
 Therefore the drift of $\tfrac{\cW^{1-\gamma}}{1-\gamma} e^{Y^p} +\int_0^\cdot f(c_s, \tfrac{\cW_s^{1-\gamma}}{1-\gamma} e^{Y^p_s})ds$ reads (after suppressing the subscript $t$)
 \[
 \begin{split}
  &\tfrac{\cW^{1-\gamma}}{1-\gamma} e^{Y^p} \Big\{(1-\gamma) r -\delta \theta + \tfrac12 |Z^p|^2 + (1-\gamma)\big[- \overline{c} + \delta e^{-\tfrac{1}{\theta} Y^p} \tfrac{1}{1-\tfrac{1}{\psi}} \overline{c}^{1-\tfrac{1}{\psi}}\big]\\
  & \hspace{2cm} + (1-\gamma) \big[-\tfrac{\gamma}{2} \pi' \Sigma \pi + \pi' (\mu + \sigma \rho (Z^p)')\big] -H(Y^p, Z^p)\Big\}.
 \end{split}
 \]
 The martingale principle then yields the previous drift to be non-positive, leading to the maximization problem in \eqref{eq:H-rep}, whose maximizer is obtained by calculation.

 The minimization problem in \eqref{eq:H-rep} is obtained after Lemma \ref{lem:mart-prin}. The dynamics of $D^*$ follows from plugging $(\xi^*, \eta^*)$ into \eqref{eq:SDE D}. It then remains to obtain the minimizer $(\xi^*, \eta^*)$. To this end, consider the unconstrained problem
 \[
  \tfrac{1}{2\gamma} (|\xi|^2 + |\eta|^2) - \tfrac{1}{\gamma} \xi z' + \lambda \sigma \rho \xi' + \lambda \sigma \rho^\bot \eta'.
 \]
 The first order condition yields
 \[
  \xi^* = z - \gamma \lambda \sigma \rho \quad \text{and} \quad \eta^* = -\gamma \lambda \sigma \rho^\bot.
 \]
 Plugging these optimizers into the constraint $\mu + \sigma \rho \xi' + \sigma \rho^\bot \eta'=0$ yields the Lagrangian multiplier $\lambda = \tfrac{1}{\gamma} (\mu' + z \rho' \sigma') \Sigma^{-1}$ and confirms $\xi^*$ and $\eta^*$ in \eqref{eq:optimizer}.
\end{proof}

\begin{proof}[Proof of Lemma \ref{lem:BSDE Y bdd}]
Since $\mu' \Sigma^{-1} \mu$ is bounded, $|\mu' \Sigma^{-1} \sigma \rho|^2 \leq \mu' \Sigma^{-1} \mu$ implies that $\tfrac{1-\gamma}{\gamma} \mu'\Sigma^{-1}\sigma \rho$ is bounded as well. Therefore, $\tfrac{d\overline{\prob}}{d\prob} = \mathcal{E}(\int \tfrac{1-\gamma}{\gamma} \mu'_s \Sigma^{-1}_s \sigma_s \rho_s dW_s)_T$\footnote{$\mathcal{E}(\int \alpha_s dW_s)_T:= \exp (-\tfrac12 \int_0^T |\alpha_s|^2 ds + \int_0^T \alpha_s dW_s)$ denotes the stochastic exponential for $\int_0^T \alpha_s dW_s$.} defines a probability measure $\overline{\prob}$ equivalent to $\prob$, hence \eqref{eq:BSDE Y} can be rewritten as
\begin{equation}\label{eq:BSDE Y bdd}
 Y_t = \int_t^T \mathcal{H}(Y_s, Z_s)ds -\int_t^T Z_s d\overline{W}_s,
\end{equation}
where $\overline{W} = W -\int_0^\cdot \tfrac{1-\gamma}{\gamma} \rho_s' \sigma_s' \Sigma^{-1}_s \mu_s ds$ is a $\overline{\prob}$-Brownian motion by the Girsanov theorem, and
\[
 \mathcal{H}(y,z) := \tfrac12 z M_t z' + \theta \tfrac{\delta^\psi}{\psi} e^{-\tfrac{\psi}{\theta} y} + h_t -\delta\theta.
\]
Here since eigenvalues of $\sigma' \Sigma^{-1} \sigma$ are either $0$ or $1$, we have $0\leq z \rho' \sigma' \Sigma^{-1} \sigma \rho z'\leq z \rho' \rho z' \leq |z|^2$. This inequality implies that
\begin{equation}\label{eq:M quad}
 0< |z|^2 \leq z M_t z' \leq \tfrac{1}{\gamma} |z|^2, \quad \text{when } 0<\gamma<1, \quad 0< \tfrac{1}{\gamma} |z|^2 \leq z M_t z' \leq |z|^2, \quad \text{when } \gamma>1.
\end{equation}
Therefore the $z$-term in $\mathcal{H}$ is positive and has quadratic growth. On the other hand, Assumption \ref{ass:bdd-coeff} implies that $h$ is bounded. We denote $h_{\min}= \essinf_{t\in[0,T]} h_t$ and $h_{\max}=\esssup_{t
\in [0,T]} h_t$. Due to the exponential term in $y$, we introduce a truncated version of \eqref{eq:BSDE Y bdd}
\begin{equation}\label{eq:BSDE Y bdd trunc}
 Y^n_t = \int_t^T \mathcal{H}^n(Y^n_s, Z^n_s) ds -\int_t^T Z^n_s d\overline{W}_s, \quad \text{ for } n>0,
\end{equation}
where the truncated generator
\[
 \mathcal{H}^n(y,z) := \tfrac12 z M_t z' + \theta \tfrac{\delta^\psi}{\psi} \big(e^{-\tfrac{\psi}{\theta}y} \wedge n\big) + h_t -\delta \theta
\]
is Lipschitz in $y$, quadratic growth in $z$, and $\mathcal{H}^n(0,0)$ is bounded. This is the \define{quadratic BSDE} studied in \cite{Kobylanski} and Theorem 2.3 therein implies that \eqref{eq:BSDE Y bdd trunc} admits a solution $(Y^n, Z^n)$ with $Y^n$ bounded and $Z^n\in H^2(\overline{\prob})$\footnote{$Z\in H^2(\overline{\prob})$ if $\expec^{\overline{\prob}}[\int_0^T |Z_s|^2 ds]<\infty$.}.

\vspace{2mm}

\noindent \textit{\underline{Case $0<\gamma<1, \gamma \psi>1$:}} The parameter specification on $\gamma$ and $\psi$ implies that $0<\theta<1$. Therefore the second term in $\mathcal{H}^n$ is positive, moreover $\mathcal{H}^n(y,z) \geq h_{\min} -\delta \theta$ for all $n$. Comparison theorem for quadratic BSDE (cf. \cite[Theorem 2.6]{Kobylanski}) yields that $Y^n_t \geq (h_{\min} -\delta \theta)(T-t)\geq -(h_{\min}-\delta \theta)_- T$, for all $t$ and $n$, where $f_- = -\min\{f, 0\}$. As a result, $\exp(-\tfrac{\psi}{\theta} Y^n) \leq \exp(\tfrac{\psi}{\theta} (h_{\min}-\delta\theta)_- T)$ for all $n$. Take $N:= \exp(\tfrac{\psi}{\theta} (h_{\min}-\delta)_- T)$. For any $n\geq N$, $\mathcal{H}(Y^n, Z^n) = \mathcal{H}^n(Y^n, Z^n)$, therefore, $(Y, Z) := (Y^n, Z^n)$ is a solution to \eqref{eq:BSDE Y bdd}.

\vspace{2mm}

\noindent \textit{\underline{Case $\gamma, \psi>1$:}} The parameter specification and $\gamma$ and $\psi$ implies that $\theta<0$. Therefore, the second term in $\mathcal{H}^n$ is negative, moreover $\mathcal{H}^n(y,z) \leq \tfrac12 |z|^2 + h_{\max} -\delta \theta$. Consider a BSDE
\[
 \overline{Y}^n_t = \int_t^T \big(\tfrac12 |\overline{Z}^n_s|^2 + h_{\max} -\delta \theta \big) ds -\int_t^T \overline{Z}^n_s d\overline{W}_s,
\]
which has the solution $\overline{Y}^n_t = (h_{\max}-\delta \theta)(T-t)$ and $\overline{Z}^n_t =0$. Then comparison theorem for quadratic BSDE yields that $Y^n_t \leq \overline{Y}^n_t\leq (h_{\max}-\delta \theta)_+ T$, for all $t$ and $n$, where $f_+ = \max\{f,0\}$. As a result, $\theta<0$ implies that $\exp(-\tfrac{\psi}{\theta} Y^n)\leq \exp(-\tfrac{\psi}{\theta}(h_{\max}-\delta \theta)_+ T)$ for all $n$. Take $N:= \exp(-\tfrac{\psi}{\theta}(h_{\max}-\delta \theta)_+ T)$. For any $n\geq N$, $\mathcal{H}(Y^n, Z^n) = \mathcal{H}^n(Y^n, Z^n)$, therefore, $(Y,Z):= (Y^n, Z^n)$ is a solution to \eqref{eq:BSDE Y bdd}.

\vspace{2mm}

Finally, we will show $Z\in H_{\BMO}$ in both cases. For any stopping time $\tau$, \eqref{eq:BSDE Y bdd} and $Z\in H^2(\overline{\prob})$ imply
\[
 \tfrac12 \expec^{\overline{\prob}}_{\tau} \Big[\int_\tau^T Z_s M_s Z_s' ds\Big] = Y_\tau - \expec^{\overline{\prob}}_{\tau}\Big[\int_\tau^T \theta \tfrac{\delta^\psi}{\psi} e^{-\tfrac{\psi}{\theta} Y_s} + h_s -\delta \theta \, ds\Big].
\]
Since $Y$ and $h$ are bounded. The right-hand side of the previous identity is bounded by some constant $C$, which does not depend on $\tau$. Therefore $\expec_\tau [\int_\tau^T Z_s M_s Z'_s ds]\leq 2C$ for any stopping time $\tau$. Combining the previous inequality with \eqref{eq:M quad}, we confirm $Z\in H_{\BMO}(\overline{\prob})$. Since $\mu' \Sigma^{-1} \sigma \rho$ is bounded, hence it also belongs to $H_{\BMO}(\prob)$. It then follows from \cite[Theorem 3.6]{Kazamaki} that $Z\in H_{\BMO}(\prob)$.
\end{proof}

\begin{proof}[Proof of Theorem \ref{thm:opt str bdd}]
 For the solution $(Y, Z)$ of \eqref{eq:BSDE Y} constructed in Lemma \ref{lem:BSDE Y bdd}, and $\pi^*, c^*, D^*$ defined in \eqref{eq:opt str}, let us define
 \begin{equation}\label{eq:UV*}
  \bU^*_t = \tfrac{1}{1-\gamma} (\mathcal{W}^{\pi^*, c^*}_t)^{1-\gamma} e^{Y_t} \quad \text{and} \quad \bV^{y*}_t = \tfrac{\gamma}{1-\gamma} (y D^*_t)^{\tfrac{\gamma-1}{\gamma}} e^{Y_t/\gamma}.
 \end{equation}
 We will prove $\bU^*, \bV^*\in \cU$, and
 \begin{align}
  & \bU^*_t = \expec_t \Big[\int_t^T f(c^*_s, \bU^*_s) ds + U_T(\mathcal{W}^{\pi^*, c^*}_T)\Big],\label{eq:U*eqn}\\
  & \bV^{y*}_t = \expec_t \Big[\int_t^T g(y D^*_s, \tfrac{1}{\gamma} \bV^{y*}_s) ds + V_T(y D^*_T)\big],\label{eq:V*eqn}
 \end{align}
 for any $y>0$ and $t\in[0,T]$. Therefore the previous statements imply that $(\pi^*, c^*)\in \mathcal{A}$ and $D^*\in \cD_a$. Take $y=y^*= w^{-\gamma} e^{Y_0}$ and denote $\bV^* = \bV^{y^**}$. We have from $\mathcal{W}^{\pi^*, c^*}_0=w$ and $D^*_0=1$ that
 \[
  \bU^*_0= \tfrac{1}{1-\gamma} w^{1-\gamma}e^{Y_0} = \tfrac{\gamma}{1-\gamma} (y^*)^{\tfrac{\gamma-1}{\gamma}} e^{Y_0/\gamma} + w y^*  = \bV^*_0 + w y^*= \inf_{y>0} (\bV_0^{y*} + wy).
 \]
 Combining the previous identity with \eqref{eq:dual-ineq}, we confirm \eqref{eq:dual-eq}.

 \vspace{2mm}

 \noindent\textit{\underline{$\bU^*\in \mathcal{U}$ and it satisfies \eqref{eq:U*eqn}:}} We denote $\mathcal{W}^{\pi^*, c^*}$ by $\mathcal{W}^*$. Using \eqref{eq:Ito primal}, where $(Y^p, Z^p)$ is replaced by $(Y,Z)$, $H$ from \eqref{eq:H-rep}, and $(\pi^*, c^*)$ from \eqref{eq:opt str}, we obtain
 \begin{align*}
  &d\pare{\mathcal{W}^*_t}^{1-\gamma} e^{Y_t} \\
  &= - \pare{\mathcal{W}^*_t}^{1-\gamma} e^{Y_t} \Big(\delta \theta \pare{c^*_s}^{1-\frac{1}{\psi}} \big(\pare{\mathcal{W}^*_t}^{1-\gamma} e^{Y_t}\big)^{-\frac{1}{\theta}}- \delta \theta\Big) dt + \pare{\mathcal{W}^*_t}^{1-\gamma} e^{Y_t} \bra{(1-\gamma) (\pi^*_t)' \sigma_t dW^\rho_t + Z_t dW_t}\\
  &= -\pare{\mathcal{W}^*_t}^{1-\gamma} e^{Y_t} \big(\theta \delta^\psi e^{-\frac{\psi}{\theta}Y_t} -\delta \theta\big) dt + \pare{\mathcal{W}^*_t}^{1-\gamma} e^{Y_t} \bra{(1-\gamma) (\pi^*_t)' \sigma_t dW^\rho_t + Z_t dW_t}.
 \end{align*}
 This implies
 \begin{equation}\label{eq:WeY}
  \pare{\mathcal{W}_t^*}^{1-\gamma} e^{Y_t} = w^{1-\gamma} e^{Y_0} \exp \Big(-\int_0^t \big(\delta^\psi \theta e^{-\frac{\psi}{\theta} Y_s} -\delta \theta\big)ds\Big) Q_t,
 \end{equation}
 where
 \begin{equation}\label{eq:Q}
 \begin{split}
  Q_t &= \mathcal{E} \Big(\int (1-\gamma) (\pi^*_s)' \sigma_s dW^\rho_s + \int Z_s dW_s\Big)_t= \mathcal{E}\Big(\int L_s dW_s + \int L^\bot_s dW_s\Big)_t,\\
  L &= \tfrac{1-\gamma}{\gamma} \mu' \Sigma^{-1} \sigma \rho + Z M, \quad L^\bot = \tfrac{1-\gamma}{\gamma} (\mu' + Z \rho' \sigma') \Sigma^{-1} \sigma \rho^\bot.
 \end{split}
 \end{equation}
 Since $Y$ is bounded, the first three terms on the right-hand side of \eqref{eq:WeY} are bounded uniformly for $t\in[0,T]$. For the exponential local martingale $Q$, note that $\mu' \Sigma^{-1} \sigma \rho \rho' \sigma' \Sigma^{-1} \mu\leq \mu' \Sigma^{-1} \mu$ and $Z M M' Z' \leq 2 [1+ (\tfrac{1-\gamma}{\gamma})^2] |Z|^2$. Therefore the boundedness of $\mu' \Sigma^{-1}\mu$ and $Z\in H_{\BMO}$ imply $L\in H_{\BMO}$ as well. Similar argument yields $L^\bot \in H_{\BMO}$. It then follows from \cite[Theorem 2.3]{Kazamaki} that $Q$ is a martingale, hence is of class (D). Coming back to \eqref{eq:WeY}, we have confirmed that $(\mathcal{W}^*)^{1-\gamma} e^{Y}$ is of class (D), hence $\bU^* \in \mathcal{U}$.

 To verify \eqref{eq:U*eqn}, we note that $\bU^* + \int_0^\cdot f(c^*_s, \bU^*_s) ds$ is a local martingale. Taking a localizing sequence $(\sigma_n)_{n\geq 1}$, we obtain
 \[
  \bU^*_t + \delta \theta \expec_t\Big[\int_t^{T\wedge \sigma_n} \bU^*_s ds\Big] = \expec_t \Big[\bU^*_{T\wedge \sigma_n} + \int_t^{T\wedge \sigma_n} \delta \frac{(c^*_s)^{1-\tfrac{1}{\psi}}}{1-\tfrac{1}{\psi}}((1-\gamma) \bU^*_s)^{1-\tfrac{1}{\theta}} ds\Big], \quad \text{ on } \{t<\sigma_n\}.
 \]
 Sending $n\rightarrow \infty$, the monotone convergence theorem and the class (D) property of $\bU^*$ yield
 \[
  \bU^*_t + \delta \theta \expec_t\Big[\int_t^{T} \bU^*_s ds\Big] = \expec_t \Big[U_T(\mathcal{W}^*_T) + \int_t^T \delta \frac{(c^*_s)^{1-\tfrac{1}{\psi}}}{1-\tfrac{1}{\psi}}((1-\gamma) \bU^*_s)^{1-\tfrac{1}{\theta}} ds\Big].
 \]
 The class (D) property of $\bU^*$ implies that $\delta \theta\expec_t[\int_t^T \bU^*_s ds]$ is finite almost surely. Subtracting it from both sides of the previous equation, we confirm \eqref{eq:U*eqn}.

 \vspace{2mm}

 \noindent\textit{\underline{$\bV^{y*}\in \mathcal{U}$ and it satisfies \eqref{eq:V*eqn}:}} Using \eqref{eq:Ito dual} together with $(\xi^*, \eta^*)$ from \eqref{eq:optimizer}, where $(Y^d, Z^d)$ is replaced by $(Y,Z)$, we obtain
 \begin{align*}
  d(D^*_t)^{\tfrac{\gamma-1}{\gamma}} e^{\tfrac{Y_t}{\gamma}} =& -\tfrac{\theta}{\gamma \psi} \delta^\psi (D^*_t)^{\tfrac{\gamma-1}{\gamma}} e^{(1-\tfrac{\gamma \psi}{\theta}) \tfrac{Y_t}{\gamma}} dt + \tfrac{\delta \theta}{\gamma} (D^*_t)^{\tfrac{\gamma-1}{\gamma}} e^{\tfrac{Y_t}{\gamma}} dt \\
  &+ (D^*_t)^{\tfrac{\gamma-1}{\gamma}} e^{\tfrac{Y_t}{\gamma}}[(1-\gamma) (\pi^*_t)' \sigma_t dW^\rho_t + Z_t dW_t]
 \end{align*}
 The previous SDE for $(D^*)^{\tfrac{\gamma-1}{\gamma}} e^{Y_t/\gamma}$ has the following solution
 \begin{equation}\label{eq:DeY}
  (D^*_t)^{\tfrac{\gamma-1}{\gamma}} e^{\tfrac{Y_t}{\gamma}} = e^{\tfrac{Y_0}{\gamma}} \exp\Big(-\tfrac{\theta}{\gamma \psi} \delta^\psi \int_0^t e^{-\tfrac{\psi}{\theta} Y_s}ds +\tfrac{\delta \theta}{\gamma}t\Big) Q_t,
 \end{equation}
 where $Q_t$ comes from \eqref{eq:Q}. Since $Y$ is bounded, the second term on the right-hand side is bounded uniformly for $t\in[0,T]$. Moreover, we have seen from the previous step that $Q$ is of class (D). Therefore $(D^*)^{\tfrac{\gamma-1}{\gamma}} e^{Y/\gamma}$ is of class (D), and so is $\bV^{y*}$ for any $y>0$. Note that $\bV^{y*} + \int_0^\cdot g(y D^*_s, \tfrac{1}{\gamma}\bV^{y*}_s) ds$ is a local martingale. The similar localization argument as the previous step confirms \eqref{eq:V*eqn}.
\end{proof}

\begin{rem}\label{rem:-Y/theta bdd}
 A careful examination reveals that the previous proof only requires $-Y/\theta$ to be bounded from above and $Q$ to be a martingale. Indeed, when $-Y/\theta$ is bounded from above, both the third term on the right-hand side of \eqref{eq:WeY} and the second term on the right-hand side of \eqref{eq:DeY} are bounded. Combined with the class (D) property of $Q$, we reach the same conclusion. We record this observation here for future reference.
\end{rem}

\begin{proof}[Proof of Corollary \ref{cor:utility graident}]
 We will prove that $D^*$ given in \eqref{eq:utility gradient} satisfies the SDE of $D^*$ in \eqref{eq:opt str}. Since this SDE clearly admits an unique solution, $D^*$ must beg given by \eqref{eq:utility gradient}. We denote $\mathcal{W}^{\pi^*, c^*}$ by $\mathcal{W}^*$ and $U^{c^*}$ by $U^*$. Combining \eqref{eq:EZ agg}, \eqref{eq:opt str} and \eqref{eq:UV*}, we obtain from calculation that
 \begin{align*}
   D^*_t &= w^{\gamma} e^{-Y_0} \exp\Big[\int_0^t \delta (\theta-1) ((1-\gamma) U^*_s)^{-\frac1\theta} (c^*_s)^{1-\frac1\psi} ds -\delta \theta t\Big] \delta ((1-\gamma) U^*_t)^{1-\frac1\theta}(c^*_t)^{-\frac1\psi}\\
  &=\exp\Big[\int_0^t (\theta-1) \delta^{\psi} e^{-\frac{\psi}{\theta} Y_s} ds - \delta \theta t\Big] \frac{(\mathcal{W}_t^*)^{-\gamma} e^{Y_t}}{w^{-\gamma} e^{Y_0}}.
 \end{align*}
 On the other hand, set $\overline{c}^* = c^*/\mathcal{W}^*$. Calculation using \eqref{eq:BSDE Y} and \eqref{eq:opt str} yield
 \begin{align*}
  d (\mathcal{W}^*)^{-\gamma} =& (\mathcal{W}^*)^{-\gamma} \Big[-\gamma(r  - \overline{c}^* + (\pi^*)' \mu) + \tfrac{\gamma(\gamma+1)}{2} (\pi^*)' \Sigma \pi^*\Big]dt - \gamma (\mathcal{W}^*)^{-\gamma} (\pi^*)' \sigma dW^\rho\\
  =& (\mathcal{W}^*)^{-\gamma} \Big[-\gamma (r -\tilde{c}^*) + \tfrac{1-\gamma}{2\gamma} \mu' \Sigma^{-1} \mu + \tfrac{1}{\gamma} \mu' \Sigma^{-1} \sigma \rho Z' + \tfrac{1+\gamma}{2\gamma} Z \rho' \sigma' \Sigma^{-1} \sigma \rho Z'\Big] dt \\
  &- \gamma (\mathcal{W}^*)^{-\gamma} (\pi^*)' \sigma dW^\rho\\
  de^Y =& e^Y \Big[-H(Y, Z) + \tfrac12 ZZ'\Big] dt + e^Y Z dW.
 \end{align*}
 Combining the previous three identities, we confirm
 \begin{align*}
  dD^* = & D^* \Big[-\gamma (r -\overline{c}^*) + (\theta -1) \delta^\psi e^{-\frac{\psi}{\theta} Y} - \delta \theta\\
  & \hspace{7mm} + \tfrac{1-\gamma}{\gamma} \mu' \Sigma^{-1}\mu + \tfrac{1-\gamma}{\gamma} \mu' \Sigma^{-1} \sigma \rho Z' + \tfrac12 Z M Z' - H(t, Y, Z)\Big] dt\\
  &+ D^* [-\gamma (\pi^*)' \sigma dW^\rho + Z dW]\\
  =& D^* \Big[-r + \big(\theta -1 - \tfrac{\theta}{\psi} + \gamma\big) \delta^\psi e^{-\frac{\psi}{\theta} Y}\Big] dt + D^* \bra{-\gamma (\pi^*)' \sigma dW^\rho + Z dW}\\
  =& -r D^* dt + D^* \bra{-\gamma (\pi^*)' \sigma dW^\rho + Z dW},
 \end{align*}
 where the third identity follows from $\theta + \gamma -1 - \frac{\theta}{\psi} = 0$.

 For the second statement, when \eqref{eq:dual-eq} holds, the first inequality in \eqref{eq:dual-bdd} must be an identity. Hence $\expec\big[\mathcal{W}^*_T D^*_T + \int_0^T D^*_s c^*_s ds\big] = w$, which implies the martingale property of $D^* \mathcal{W}^* +\int_0^\cdot D^*_s c^*_s ds$, since this process is already a supermartingale by the definition of state price density.
\end{proof}

\begin{proof}[Proof of Theorem \ref{thm:opt str unbdd}]
 Since $Y$ is bounded from above and $\theta<0$, we have $-Y/\theta$ to be bounded from above. On the other hand, \cite[Lemma B.2]{Xing} proved that $Q$ from \eqref{eq:Q} is a martingale. Therefore the statement readily follows from Remark \ref{rem:-Y/theta bdd}.
\end{proof}

\begin{proof}[Proof of Proposition \ref{prop:heston}]
 This proof is a minor generalization of \cite[Proposition 3.2]{Xing}, whose Assumption 2.11 is no longer needed here, see Remark \ref{rem:relax cond}. For the rest assumptions,
 Assumption \ref{ass:gen-coeff} follows from the fact that $r(x) + \tfrac{1}{2\gamma} \mu(x)' \Sigma(x)^{-1} \mu(x) = r_0 + (r_1 + \tfrac{1}{2 \gamma} \lambda' \Theta(x) \lambda) x$ which is bounded from below on $(0,\infty)$. Assumptions \ref{ass:over P} and \ref{ass: phi} are verified in what follows.

 \vspace{2mm}
 \noindent\underline{Assumption \ref{ass:over P}:} Note $\tfrac{1-\gamma}{\gamma} \mu(x)' \Sigma(x)^{-1} \sigma(x) \rho(x) = \tfrac{1-\gamma}{\gamma} \lambda' \Theta(x) \rho \sqrt{x}$. Consider the martingale problem associated to $\overline{\mathcal{L}}:= \bra{b\ell - \pare{b-\frac{1-\gamma}{\gamma} a \lambda' \Theta(x)\rho} x}\partial_x + \frac12 a^2 x \partial^2_x$ on $(0,\infty)$. Since $\Theta(x)$ is bounded and $b\ell > \frac12 a^2$, Feller's test of explosion implies that the previous martingale problem is well-posed. Then \cite[Remark 2.6]{Cheridito-Filipovic-Yor} implies that the stochastic exponential in Assumption \ref{ass:over P} (i) is a $\prob-$martingale, hence $\overline{\prob}$ is well defined. For Assumption \ref{ass:over P} (ii), $h(x) = (1-\gamma) r_0 + \bra{(1-\gamma) r_1 + \frac{1-\gamma}{2\gamma} \lambda'\Theta\lambda}x$. Since $X$ has the following dynamics under $\overline{\prob}$:
  \[
   dX_t = \Big[b\ell - \Big(b- \tfrac{1-\gamma}{\gamma}a \lambda' \Theta(x)\rho\Big)X_t\Big] + a \sqrt{X_t} d\overline{W}_t,
  \]
  where $\overline{W}$ is a $\overline{\prob}-$Brownian motion. Then $\expec^{\overline{\prob}}[\int_0^T h(X_s) ds]>-\infty$ follows from the fact that $\Theta(x)$ is bounded hence $\expec^{\overline{\prob}}[X_s]$ is bounded uniformly for $s\in[0,T]$.

  \vspace{2mm}

  \noindent{\underline{Assumption \ref{ass: phi}}:} The operator $\fF$ in \eqref{eq: fF} reads
  \[
   \fF[\phi] = \tfrac12 a^2 x \partial^2_x \phi + \Big(b\ell - bx + \tfrac{1-\gamma}{\gamma} a \lambda' \Theta(x)\rho x\Big) \partial_x \phi +\tfrac12 \tilde{M} a^2 x (\partial_x \phi)^2 + (1-\gamma) (r_0+ r_1 x) + \tfrac{1-\gamma}{2\gamma} \lambda'\Theta(x)\lambda x,
  \]
  where $\tilde{M} = 1+ \tfrac{1-\gamma}{\gamma} \rho' \Theta(x)\rho>0$. Consider $\phi(x) = -\underline{c} \log x + \overline{c} x$, for two positive constants $\underline{c}$ and $\overline{c}$ determined later. It is clear that $\phi(x)\uparrow \infty$ when $x\downarrow 0$ or $x\uparrow \infty$. On the other hand, calculation shows
  \begin{align*}
   \fF[\phi] =& C +  \Big[\tfrac12 a^2\underline{c} + \tfrac12 a^2\underline{c}^2 \tilde{M} - b\ell \underline{c}\Big] \frac1x \\
   &+ \Big[-\Big(b- \tfrac{1-\gamma}{\gamma} a \lambda' \Theta(x)\rho\Big) \overline{c} + \tfrac12 a^2 \overline{c}^2 \tilde{M} + (1-\gamma) r_1 + \tfrac{1-\gamma}{2\gamma} \lambda' \Theta(x)\lambda\Big]x,
  \end{align*}
  where $C$ is a constant. Since $b \ell >\frac12 a^2$, the coefficient of $1/x$ is negative for sufficiently small $\underline{c}$. When $r_1$ or $\lambda'\Theta(x)\lambda>0$, since $\gamma>1$ and $\Theta(x)$ is bounded, the coefficient of $x$ is negative for sufficiently small $\overline{c}$. Therefore, these choices of $\underline{c}$ and $\overline{c}$ imply that $\fF[\phi](x) \downarrow -\infty$ when $x\downarrow 0$ or $x\uparrow \infty$, hence $\fF[\phi]$ is bounded from above on $(0,\infty)$, verifying Assumption \ref{ass: phi}.
\end{proof}

\bibliographystyle{abbrvnat}
\bibliography{biblio}
\end{document}